\documentclass[twocolumn]{autart}    

\usepackage{graphicx}          
\usepackage{xcolor}
\usepackage{amsmath,amsfonts,amssymb}
\usepackage{soul} 
\definecolor{ThananaEdit}{rgb}{0.0,0.3,0.9} 
\definecolor{ThananaMark}{rgb}{0.80,0.90,1.00} 

\usepackage{tikz}
\usetikzlibrary{arrows.meta,positioning,shapes.geometric}
\usetikzlibrary{shadows}
\usetikzlibrary{shadings}
\usetikzlibrary{backgrounds, fit}

\usepackage{algorithm}
\usepackage{algpseudocode}

\usepackage{multirow}

\usepackage{enumitem}

\usepackage{mathrsfs}

\usepackage{graphicx}

\begin{document}
	
	\begin{frontmatter}
		
		\title{Cognitive Flexibility as a Latent Structural Operator for Bayesian State Estimation} 
		
		\author[KMUTT]{Thanana Nuchkrua}\ead{thanana.nuch@yahoo.com}%
		\author[KMUTT]{Sudchai Boonto}\ead{sudchai.boo@kmutt.ac.th}%
		\author[UIC]{Xiaoqi Liu}\ead{xliu276@uic.edu}%
		
		\address[KMUTT]{Department of Control Systems and Instrumentation Engineering, King Mongkut's University of Technology Thonburi, Thailand}  
		\address[UIC]{Department of Computer Science, University of Illinois Chicago, USA.}             


		\begin{abstract}                          
			Deep stochastic state-space models enable Bayesian filtering in nonlinear,
			partially observed systems but typically assume a fixed latent structure.
			When this assumption is violated, parameter adaptation alone may result in
			persistent belief inconsistency.
			We introduce \emph{Cognitive Flexibility} (CF) as a representation-level
			operator that selects latent structures online via an innovation--based
			predictive score, while preserving the Bayesian filtering recursion.
			Structural mismatch is formalized as irreducible predictive inconsistency
			under fixed structure. The resulting belief--structure recursion is shown
			to be well posed, to exhibit a structural descent property, and to admit
			finite switching, with reduction to standard Bayesian filtering under
			correct specification.
			Experiments on latent-dynamics mismatch, observation-structure shifts, and
			well-specified regimes confirm that CF improves predictive accuracy under a
			mismatch while remaining non-intrusive when the model is correctly specified.
		\end{abstract}
		
	\end{frontmatter}
	\section{Introduction}

Modern learning-enabled control systems~\cite{11312267,8103164} increasingly operate in environments
	where the relationship between system states, observations, and inputs is not
	fixed, but evolves over time.
	Such evolution arises in many physical systems~\cite{Lusch2018UniversalLinearEmbeddings} due to changes in sensing
	modalities, operating regimes~\cite{10551376}, task semantics, or interaction conditions, and is
	particularly pronounced in systems with compliant dynamics~\cite{Nuchkrua2022TIE} or strong
	environmental coupling
	\cite{Derler2012CPS,Lilge2022ContinuumStateEstimation}.
	When these changes occur, a model that is locally accurate can become globally
	misaligned with the true data-generating process, leading to persistent
	prediction errors and degraded closed-loop performance—even when classical
	parameter adaptation or robustification techniques are employed
	\cite{Ljung1999system,Qin2003Survey}.
	Understanding how to reason about and respond to such
	\emph{structural nonstationarity} is therefore central to reliable control and
	decision-making under uncertainty.

 In general, uncertainty in control and decision-making is addressed by assuming a
	\emph{fixed} model structure and compensating for mismatch through parameter
	adaptation, robust control design, or stochastic noise modeling
	\cite{jazwinski1970stochastic,maybeck1979stochastic,Rawlings2017MPC}.
	Under this paradigm, control and prediction are carried out with respect to a
	\emph{state belief}—the inferred distribution over latent states given available
	measurements—rather than the true, unobserved system state~\cite{Kaelbling1998POMDP}.
	Bayesian state estimation  \cite{Thrun2005ProbabilisticRobotics} then provides a coherent mechanism for the time
	evolution of this belief and forms the backbone of learning-enabled
	control.

However, when the assumed latent structure itself is incorrect, these mechanisms
are fundamentally limited: the resulting belief can remain numerically
well-defined while becoming systematically inconsistent with the true system
behavior.
This phenomenon—here termed \emph{structural mismatch}—cannot be eliminated by
parameter updates alone and constitutes an intrinsic failure mode of fixed
representation models.
Despite its practical relevance across robotics, autonomous systems
and learning--based control, structural mismatch has
received limited formal treatment \emph{at the level of Bayesian belief
	evolution itself} (i.e.,~\cite{Hewing2020,Becker2021,curi2020efficientmodelbasedreinforcementlearning}).

In recent years, data-driven modeling has significantly extended the classical
state-space model (SSM) framework~\cite{11107969}.
In particular, \emph{Deep Stochastic State-Space Models} (DeepSSSMs)~\cite{gedon2021deepssm,lin2024deeplearningbasedapproachesstate} combine
Bayesian filtering with expressive nonlinear representations learned from data,
\emph{enabling state estimation and prediction} in complex and high-dimensional
systems, including vision--based and latent-dynamics models for planning and
control \cite{krishnan2015deep,Fraccaro2017,Karl2017,Hafner2019,durkan2020neural}.
Beyond their origins in sequence modeling, deep state-space formulations have
increasingly been adopted in system identification and control-oriented
modeling, including neural state-space architectures, encoder--based
identification pipelines, and stochastic latent models for learning--based
control \cite{gedon2021deepssm,forgione2021dynonet,beintema2021deepenc,Becker2021,Soloperto2023BayesianActuation,lin2024review}.
Despite this progress, most DeepSSSM formulations retain a key assumption
inherited from classical models: the \emph{latent structure} of the state-space
model is fixed throughout operation.

This fixed-structure assumption becomes restrictive precisely in the regimes
where learned models are most attractive: deployment under changing sensing and
interaction conditions, and operation beyond the training distribution~\cite{Ovadia2019CanYouTrust}.
In practice, the relationship between latent states and observations may change
due to sensor degradation, environmental variation, unmodeled operating
regimes, or shifts in task semantics (i.e.,~\cite{diehl2022umbrellauncertaintyawaremodelbasedoffline}).
When such changes occur, parameter adaptation within a fixed latent
representation is often insufficient: the Bayesian belief can remain numerically
well-defined while becoming systematically misaligned with the true
data-generating process, producing persistent prediction errors and degraded
closed-loop performance \cite{Ljung1999system,Slotine1991,Ioannou1996}.
This issue is particularly acute in settings where uncertainty quantification,
risk sensitivity, and reliability are central to safe decision-making
\cite{Aswani2013,Chow2018,Berkenkamp2021SafeRL,Thananjeyan2021Safety}.

The need to address model mismatch and nonstationarity has long been recognized
in control and estimation~\cite{10077740,pnas2218197120,Pillonetto2022}.
Classical approaches include adaptive observers~\cite{besancon2007nonlinear}, gain scheduling~\cite{oliveira2025conservativeadaptivegainschedulingcontrol}, and
multiple-model estimation \cite{barshalom2001estimation,narendra1989stable}.
Interacting multiple-model (IMM) filters and hybrid observers~\cite{BALLUCHI2013915,KONG2021109752} allow transitions
among a finite set of pre-specified structures and admit strong theoretical
guarantees when the relevant operating regimes can be identified \emph{a priori}
\cite{blom1988interacting,barshalom2001estimation,Lavaei2022SurveySHS,10638633}.
These methods clarify an important point: \emph{structural change can be handled},
but typically only when one can enumerate the ``right'' modes in advance and
maintain mode-consistent filtering models.

In many contemporary data-driven settings, however, the enumeration assumption
underlying classical hybrid and multiple-model approaches is difficult to sustain.
Structural mismatch may not be well captured by a small, fixed bank of candidate
models, and learned latent representations can fail in ways that are not easily
diagnosed by standard residual analysis or noise inflation.
Recent work has therefore explored learning-enhanced filtering pipelines
\cite{revach2022kalmannet}, meta-learning strategies
\cite{chakrabarty2023metalearning,McClement2022MetaRL}, and cross-task
generalization \cite{lin2024review}.
While these approaches substantially expand representational capacity, they leave
open a system-theoretic question that is central to reliability:
\emph{how should Bayesian belief evolution respond when the latent representation
	itself becomes restrictive?}

We introduce \emph{Cognitive Flexibility} (CF)
	\cite{Scott1962CognitiveFlexibility,Collins2013} as a belief-level mechanism for
	\emph{structural reorganization} in DeepSSSMs.
	CF is formulated as an operator that selects which latent representation governs
	belief evolution at a given time.
	For any fixed structure, the underlying Bayesian filtering recursion is left
	unchanged; CF acts solely by enabling controlled transitions among
	representations when persistent belief inconsistency indicates that the current
	structure has become restrictive.
	As a result, representation adaptation is made explicit and analyzable, while
	preserving the probabilistic well-posedness of belief evolution.

Accordingly, CF is not an estimation heuristic but a
	representation-level \emph{control variable} governing belief evolution under
	structural nonstationarity, operating over a predefined family of latent
	structures rather than synthesizing new representations online.

From a system-theoretic perspective, this formulation raises three questions not
explicitly addressed by existing DeepSSSM or hybrid-estimation frameworks:
(i) how to characterize structural mismatch as an intrinsic limitation of fixed
latent representations;
(ii) how to model representation reorganization as an operator that interacts
with, rather than replaces, Bayesian filtering; and
(iii) under what conditions online structural adaptation can improve predictive
consistency while remaining controlled and well posed.

\textbf{Contributions.}
This paper advances a belief-level perspective on representation adaptation and
its system-theoretic implications.
The main contributions are as follows.

\textbf{(i) Structural mismatch as a fundamental estimation failure mode.}
We formalize \emph{structural mismatch} as an irreducible divergence between the
true conditional state distribution and the posterior belief induced by any
\emph{fixed} latent structure.
This characterization identifies a class of estimation errors that cannot be
eliminated by parameter adaptation, robustification, or noise modeling alone
\cite{Ljung1999system,narendra1989stable,Ioannou1996}.

\textbf{(ii) Cognitive Flexibility as a belief-level structural operator.}
We introduce \emph{Cognitive Flexibility} (CF) as a latent structural operator
coupled directly to Bayesian filtering recursions.
In contrast to classical and learning--based state--space models that assume a
fixed latent representation and adapt only through parameter updates
\cite{jazwinski1970stochastic,maybeck1979stochastic,Fraccaro2017,Karl2017,Hafner2019,gedon2021deepssm},
CF enables regulated transitions across latent structures. 

\textbf{(iii) System-theoretic properties of adaptive belief evolution.}
We establish fundamental properties of the resulting belief--structure dynamics,
including invariance of the belief space, monotone innovation--based structural
improvement, finite switching under persistent score separation, and reduction
to standard Bayesian filtering under correct structural specification.
These results complement classical multiple-model and hybrid estimation
frameworks \cite{blom1988interacting,barshalom2001estimation} by providing a
belief-level characterization of representation reorganization and clarifying
when structural adaptation is beneficial versus non-intrusive.

Numerical experiments demonstrate recovery from latent-dynamics mismatch,
adaptation under observation-structure shifts, and non-intrusiveness in
well-specified regimes.

\noindent\textbf{Relevance to control.}
The belief $\mathfrak{B}_t$ produced by the CF--augmented
	filter serves directly as the information state for
	belief-space control laws~\cite{jazwinski1970stochastic,Bertsekas2005v2},
	including MPC schemes that plan over the predictive
	distribution~\cite{Hewing2020}.
	Structural mismatch---the failure mode formalized in
	Theorem~10---propagates directly to control performance:
	a misspecified belief inflates uncertainty estimates,
	induces overly conservative constraint tightening, and
	degrades closed-loop tracking.
	CF addresses this failure at the belief level,
	before it reaches the control layer.
	A companion paper~\cite{Nuchkrua2026lcss} develops the
	corresponding robust CF theory for noisy innovation scores,
	connecting the present estimation framework to practical
	control implementations.  

The remainder of the paper is organized as follows.
Section~\ref{sec:formulation} introduces the problem formulation and belief
representation.
Section~\ref{sec:CFDeeepSSSM} presents the CF framework as a structural operator
on the belief space.
Sections~\ref{sec:layer1}--\ref{sec:layer3} analyze well-posedness, structural
descent, finite switching, and long-run behavior.
Section~\ref{sec:experiments} reports numerical studies, and
Section~\ref{sec:conclusion} concludes with implications and future directions.

\subsection{Notation}
All random variables are defined on a complete
probability space $(\Omega,\mathcal{F},\mathbb{P})$.
Time is discrete with
$t\in\mathbb{N}:=\{0,1,2,\dots\}$.
Let $u_t\in\mathcal{U}$ denote a known input and
$y_t\in\mathcal{Y}$ the corresponding measurement.
The latent state, observation, and input processes
are $\{z_t\}_{t\ge0}$, $\{y_t\}_{t\ge0}$,
and $\{u_t\}_{t\ge0}$.
Process and measurement noises satisfy
$w_t \sim \mathcal{W}(\cdot\mid z_t,u_t)$ and
$v_t \sim \mathcal{V}(\cdot\mid z_t)$
with variances $\sigma_w^2$ and $\sigma_v^2$.
Let $\mathcal{Z}$ be a Polish space and
$\mathcal{P}(\mathcal{Z})$ the set of Borel
probability measures on $\mathcal{Z}$.
If $\mu\in\mathcal{P}(\mathcal{Z})$ admits a density,
we identify $\mu$ with its density.
Expectation under $\mu$ is $\mathbb{E}_{\mu}[\cdot]$,
and $D_{\mathcal{KL}}(\mu\|\nu)$ denotes the
Kullback--Leibler divergence.
The information $\sigma$-algebra at time $t$ is
$\mathcal{I}_t := \sigma(y_{1:t},u_{1:t-1})$.
The posterior belief is
$\mathfrak{B}_t \in \mathcal{P}(\mathcal{Z})$.
A latent structure is indexed by $s\in\mathcal{S}$,
where $\mathcal{S}$ is finite; the active structure
$s_t\in\mathcal{S}$ is a deterministic function
of $\mathcal{I}_t$.
Let $\theta \in \Theta \subset \mathbb{R}^p$
denote a parameter vector.
The innovation likelihood is
$\ell_{\theta,s}(y_{t+1}\mid\mathfrak{B}_t,u_t)
:= \int p_{\theta,s}(y_{t+1}\mid z)\,
(\mathcal{P}_{\theta,s}\mathfrak{B}_t)(dz)$.
Let $\mathcal{F}_{\theta}:
\mathcal{P}(\mathcal{Z})\times\mathcal{U}
\times\mathcal{Y}\to\mathcal{P}(\mathcal{Z})$
denote the Bayesian filtering operator and
$\mathcal{F}_{\theta,s}$ its restriction to
structure $s$.
The constant $\gamma\in(0,1]$ denotes a structural
separation parameter.
	
\section{Preliminaries and Problem Formulation}
\label{sec:Problem_formulation}

We consider discrete-time state estimation under
partial observations, where both the state evolution
and observation process are subject to stochastic
disturbances and may change over time. The central
challenge is that no single fixed model may
consistently describe the system behavior across
all operating conditions — a limitation that
motivates the CF framework developed below.

\subsection{Preliminaries}
\label{sec:preliminaries}

The physical process is described abstractly as
\begin{align}
	z_{t+1} &= f(z_t, u_t, w_t),
	\label{eq:plant_state}\\
	y_t     &= h(z_t, v_t),
	\label{eq:plant_obs}
\end{align}
where $f:\mathcal{Z}\times\mathcal{U}\to\mathcal{Z}$
and $h:\mathcal{Z}\to\mathcal{Y}$ are unknown and
possibly time-varying, reflecting modeling uncertainty
and changes in operating conditions.
The CF framework developed here complements a
companion control
application~\cite{nuchkrua2026cognitiveflexiblecontrollatentmodel},
in which CF governs belief evolution within a
predictive safety control architecture.

\begin{rem}[Modeling scope]
	\label{rem:remark1}
	We do not assume that $(f, h)$
	in~\eqref{eq:plant_state}--\eqref{eq:plant_obs}
	belong to any prescribed model class.
	In particular, we do not impose $f \in \mathcal{F}_0$
	and $h \in \mathcal{H}_0$ for given hypothesis classes
	\begin{equation*}
		\mathcal{F}_0 \subset
		\{f:\mathcal{Z}\times\mathcal{U}\to\mathcal{Z}\},
		\qquad
		\mathcal{H}_0 \subset
		\{h:\mathcal{Z}\to\mathcal{Y}\}.
	\end{equation*}
	The data-generating mechanism may satisfy
	$(f,h)\notin\mathcal{F}_0\times\mathcal{H}_0$,
	inducing \emph{structural mismatch}: inference
	is performed under a misspecified model class, so
	that even optimal parameter adaptation within
	$\mathcal{F}_0\times\mathcal{H}_0$ cannot restore
	predictive consistency, resulting in persistent
	estimation
	error~\cite{Ljung1999system,2025MisspecifiedMLE}.
\end{rem}

\subsection{Problem formulation}
\label{sec:formulation}

Rather than committing to a potentially misspecified
structural model in~\eqref{eq:plant_state}--%
\eqref{eq:plant_obs}, we formulate inference
directly at the level of conditional probability
laws~\cite{jazwinski1970stochastic,anderson1979optimal}.
The following development is necessarily detailed
because the latent structure $s$ enters at three
distinct levels — the model class, the filtering
operator, and the belief trajectory — each of which
must be distinguished to state the main results of
Section~\ref{sec:CFDeeepSSSM} precisely.
The central object is the \emph{posterior belief}
\begin{equation}
	\mathfrak{B}_t(\cdot)
	:= \mathbb{P}\!\left(z_t \in \cdot \mid
	\mathcal{I}_t\right)
	\in \mathcal{P}(\mathcal{Z}),
	\label{eq:belief_space}
\end{equation}
i.e., the conditional law of $z_t$ given
$\mathcal{I}_t := \sigma(y_{1:t}, u_{1:t-1})$.
The belief $\mathfrak{B}_t$ is a sufficient statistic
for Bayesian state estimation~\cite{maybeck1979stochastic}:
all inference about $z_t$ conditioned on
$\mathcal{I}_t$ can be expressed through
$\mathfrak{B}_t$, which absorbs uncertainty from
$u_t$, $w_t$, and $v_t$
in~\eqref{eq:plant_state}--\eqref{eq:plant_obs}.
In particular, $\mathfrak{B}_t$ is an information
state: any conditional quantity of interest —
state predictions, conditional expectations, or
control-relevant functionals
$J:\mathcal{P}(\mathcal{Z})\to\mathbb{R}$ —
depends on $(y_{1:t}, u_{1:t-1})$ only through
$\mathfrak{B}_t$~\cite{jazwinski1970stochastic,anderson1979optimal}.
When $\mathbb{P}(z_t\in\cdot\mid y_{1:t},u_{1:t-1})$
admits a Lebesgue density, $\mathfrak{B}_t$ takes
the pointwise form
\begin{equation}
	\mathfrak{B}_t(z)
	= p(z_t = z \mid y_{1:t},\,u_{1:t-1}),
	\label{eq:belief_def}
\end{equation}
which we use interchangeably with the
measure-valued formulation~\eqref{eq:belief_space}.

In the DeepSSSM
framework~\cite{nuchkrua2026cognitiveflexiblecontrollatentmodel},
the abstract maps $(f, h)$
in~\eqref{eq:plant_state}--\eqref{eq:plant_obs}
are not identified directly.
Although the notation follows this framework,
the results of Section~\ref{sec:CFDeeepSSSM}
apply to any parameterised Bayesian filter of
the form~\eqref{eq:filter_map_s}, independently
of the specific architecture used to
represent $p_\theta$.
Instead, as noted in Remark~\ref{rem:remark1},
their effect on belief evolution is captured
through a parameterised family of conditional
distributions:
\begin{align}
	z_{t+1} &\sim p_{\theta}(z_{t+1}\mid z_t, u_t),
	\label{eq:model_trans}\\
	y_t      &\sim p_{\theta}(y_t \mid z_t),
	\label{eq:model_emis}
\end{align}
where $\theta$ is learned from data.
The model class~\eqref{eq:model_trans}--%
\eqref{eq:model_emis} induces a Bayesian filtering
recursion on $\mathcal{P}(\mathcal{Z})$,
\begin{equation}
	\underbrace{\mathfrak{B}_{t+1}}_{\text{updated belief}}
	=
	\underbrace{\mathcal{F}_{\theta}}_{\text{filtering operator}}
	\!\Big(
	\underbrace{\mathfrak{B}_t}_{\text{current belief}},\,
	\underbrace{u_t,\,y_{t+1}}_{\text{data}}
	\Big),
	\label{eq:filter_map}
\end{equation}
where $\mathcal{F}_\theta : \mathcal{P}(\mathcal{Z})
\times \mathcal{U} \times \mathcal{Y}
\to \mathcal{P}(\mathcal{Z})$ is the standard
Bayesian filtering
operator~\cite{maybeck1979stochastic}.
For fixed $\theta$, \eqref{eq:filter_map} defines
a deterministic dynamical system on
$\mathcal{P}(\mathcal{Z})$, driven by
$(u_t, y_{t+1})$.

Equation~\eqref{eq:filter_map} implicitly assumes
a fixed model structure: inference adapts only
the parameterisation $\theta$ within a prescribed
model class. This assumption breaks down when
$\mathfrak{B}_t$ also depends on a
\emph{latent structure} $s\in\mathcal{S}$ that
specifies the model class
itself.\footnote{A constructive realization and
	examples of $\mathcal{S}$ are developed
	in~\cite{nuchkrua2026cognitiveflexiblecontrollatentmodel}.}
Formally, for each $s\in\mathcal{S}$,
\begin{equation*}
	\mathcal{Z}_s \subseteq \mathcal{Z},\quad
	p_{\theta,s}: \mathcal{Z}_s \times \mathcal{U}
	\to \mathcal{P}(\mathcal{Z}_s),\quad
	q_{\theta,s}: \mathcal{Z}_s \to
	\mathcal{P}(\mathcal{Y}),
\end{equation*}
with $z_{t+1}\mid z_t,u_t \sim
p_{\theta,s}(\cdot\mid z_t,u_t)$ and
$y_t\mid z_t \sim q_{\theta,s}(\cdot\mid z_t)$,
leading to structure-dependent belief dynamics.

Remark~\ref{rem:remark1} identifies the
possibility of structural mismatch at the level
of $(f,h)$; the following definition makes this
precise at the level of the filtering operator
by restricting $\mathcal{F}_\theta$ to the model
class induced by a fixed $s\in\mathcal{S}$.

\begin{defn}[Belief dynamics under structure $s$]
	\label{def:fixed_structure_dynamics}
	Under $s\in\mathcal{S}$, the belief evolves via
	$\mathcal{F}_\theta$ restricted to the model
	class induced by $s$:
	\begin{equation}
		\mathfrak{B}_{t+1}
		=
		\underbrace{\mathcal{F}_{\theta,s}}_{\substack{
				\text{structure-restricted}\\
				\text{filter of }\mathcal{F}_{\theta}}}
		\!\Big(\mathfrak{B}_t, u_t, y_{t+1}\Big).
		\label{eq:filter_map_s}
	\end{equation}
\end{defn}

For a fixed $s\in\mathcal{S}$, the general
recursion~\eqref{eq:filter_map} thus reduces to
the structure-conditioned
update~\eqref{eq:filter_map_s}, restricting
inference to the associated model class.
The central difficulty arises when the true
latent dynamics in~\eqref{eq:model_trans} lie
outside this class: belief propagation
via~\eqref{eq:filter_map_s} remains well posed
but becomes misspecified, producing persistent
innovation errors and degraded predictive
performance. This is the regime of
\emph{structural mismatch} that CF is designed
to address.

\subsection{Problem Statement}
\label{sec:problem_statement}

The analysis of Section~\ref{sec:formulation} reveals
a fundamental limitation: when the true dynamics
lie outside the model class induced by any fixed
$s\in\mathcal{S}$, no parameter adaptation within
that class can restore predictive consistency.
This motivates a mechanism that treats the latent
structure $s_t$ as a degree of freedom to be selected
online, rather than a fixed modelling choice.

Specifically, the problem is to design an estimation
mechanism that jointly updates the belief
$\mathfrak{B}_t$ and the active structure
$s_t\in\mathcal{S}$ at each time step.
We consider a joint belief--structure recursion of
the form
\begin{equation}
	(\mathfrak{B}_t,\,s_t)
	\;\mapsto\;
	(\mathfrak{B}_{t+1},\,s_{t+1}),
	\label{eq:joint_update}
\end{equation}
where $\mathfrak{B}_{t+1}$ is propagated under the
selected structure $s_{t+1}$
via~\eqref{eq:filter_map_s}.
The key requirement is that the structural update
$s_t \mapsto s_{t+1}$ be driven by evidence of
predictive inconsistency — so that CF intervenes
only when the current structure has become
restrictive — while the Bayesian recursion
itself remains unchanged.
	
\section{Cognitive Flexibility as a Latent Structural 
	Operator}
\label{sec:CFDeeepSSSM}

Section~\ref{sec:Problem_formulation} establishes that
structural mismatch is an intrinsic limitation of
fixed-structure belief evolution: no parameter
adaptation within a fixed $s\in\mathcal{S}$ can restore
predictive consistency once the true dynamics lie outside
the induced model class.
Cognitive Flexibility (CF) resolves this by treating
$s_t$ as a representation-level variable updated online
alongside $\mathfrak{B}_t$, while leaving the Bayesian
recursion unchanged.
CF operates on the coupled state
$(\mathfrak{B}_t,s_t)\in\mathcal{P}(\mathcal{Z})
\times\mathcal{S}$ through two components:
belief evolution on $\mathcal{P}(\mathcal{Z})$ under
fixed $s$, and innovation-driven structural adaptation
on $\mathcal{S}$; see Fig.~\ref{fig:cf_latent_operator}.
The analysis proceeds in three layers:
well-posedness and fixed-structure limitations
(Section~\ref{sec:layer1}), the structural adaptation
mechanism (Section~\ref{sec:layer2}), and asymptotic
behavioral consequences (Section~\ref{sec:layer3}).

	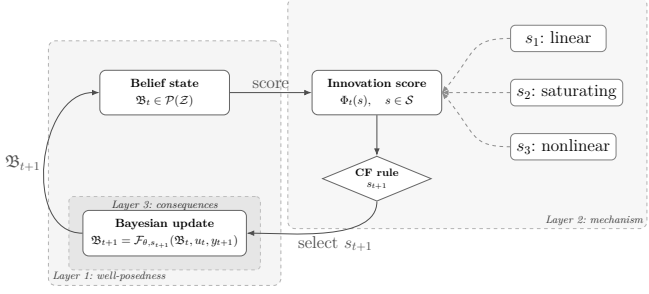
\begin{figure}[t]
		\centering
		\resizebox{0.48\textwidth}{!}{%
			\begin{tikzpicture}[
				font=\normalsize,
				>=Latex,
				node distance=12mm,
				box/.style={
					draw=black!60,
					rounded corners=5pt,
					align=center,
					inner sep=7pt,
					line width=0.6pt,
					fill=white,
					minimum width=38mm
				},
				sbox/.style={
					draw=black!50,
					rounded corners=4pt,
					align=center,
					inner sep=5pt,
					line width=0.5pt,
					fill=white,
					minimum width=28mm
				},
				decision/.style={
					draw=black!60,
					diamond,
					aspect=2.4,
					align=center,
					inner sep=2pt,
					line width=0.6pt,
					fill=white
				},
				arr/.style={
					->, line width=0.7pt, black!70
				},
				darr/.style={
					->, dashed, line width=0.5pt, black!45
				}
				]
				
				\node[box] (belief)
				{\textbf{Belief state}\\[3pt]
					$\mathfrak{B}_t \in \mathcal{P}(\mathcal{Z})$};
				
				\node[box, below=28mm of belief] (update)
				{\textbf{Bayesian update}\\[3pt]
					$\mathfrak{B}_{t+1} =
					\mathcal{F}_{\theta,s_{t+1}}
					(\mathfrak{B}_t,u_t,y_{t+1})$};
				
				\node[box, right=24mm of belief] (score)
				{\textbf{Innovation score}\\[3pt]
					$\Phi_t(s),\quad s\in\mathcal{S}$};
				
				\node[decision, below=12mm of score] (cf)
				{\small\textbf{CF rule}\\[-2pt]
					\small$s_{t+1}$};
				
				\node[sbox, right=20mm of score, text=black!100, 
				yshift= 16mm] (s1)
				{\large{$s_1$: linear}};
				
				\node[sbox, right=20mm of score] (s2)
				{\large{$s_2$: saturating}};
				
				\node[sbox, right=20mm of score,
				yshift=-16mm] (s3)
				{\large{$s_3$: nonlinear}};
				
				\draw[arr] (belief) --
				node[midway, above, font=\large]{score}
				(score);
				
				\draw[arr] (score) -- (cf);
				
				\draw[arr] (cf.south)
				.. controls +(0,-12mm) and +(12mm,0) ..
				node[pos=0.45, below, font=\large]
				{select $s_{t+1}$}
				(update.east);
				
				\draw[arr] (update.west)
				.. controls +(-18mm,0) and +(-18mm,0) ..
				node[pos=0.50, left, font=\large]
				{$\mathfrak{B}_{t+1}$}
				(belief.west);
				
				\draw[darr] (s1.west)
				.. controls +(-12mm,0) and +(10mm,12mm) ..
				(score.east);
				
				\draw[darr] (s2.west) -- (score.east);
				
				\draw[darr] (s3.west)
				.. controls +(-12mm,0) and +(10mm,-12mm) ..
				(score.east);
				
				\begin{scope}[on background layer]
					
					\node[draw=black!30, dashed,
					rounded corners=6pt,
					inner xsep=10mm, inner ysep=8.6mm,
					fill=black!4, line width=0.5pt,
					fit=(belief)(update),
					label={[font=\small\itshape,
						text=black!80,
						anchor=south west]
						south west:Layer 1: well-posedness}]
					(L1) {};
					
					\node[draw=black!30, dashed,
					rounded corners=6pt,
					inner xsep=7mm, inner ysep=8mm,
					fill=black!3, line width=0.5pt,
					fit=(score)(cf)(s1)(s2)(s3),
					label={[font=\small\itshape,
						text=black!80,
						anchor=south east]
						south east:Layer 2: mechanism}]
					(L2) {};
					
					\node[draw=black!25, dashed,
					rounded corners=4pt,
					inner xsep=4mm, inner ysep=4mm,
					fill=black!10, line width=0.4pt,
					fit=(update),
					label={[font=\small\itshape,
						text=black!85,
						anchor=north]
						north:Layer 3: consequences}]
					(L3) {};
					
				\end{scope}
				
			\end{tikzpicture}%
		}
		\caption{%
			The CF pipeline as a latent structural operator.
			At each step, the innovation scores
			$\{\Phi_t(s)\}_{s\in\mathcal{S}}$ are evaluated
			against the current belief $\mathfrak{B}_t$ and
			passed to the CF rule~\eqref{eq:cf_select},
			which selects $s_{t+1}$ and parameterises the
			Bayesian update~\eqref{eq:cf_belief_update}.
			Dashed regions correspond to the three analytical
			layers of Section~\ref{sec:CFDeeepSSSM}.
		}
		\label{fig:cf_latent_operator}
	\end{figure}


\begin{assum}[Fixed latent structure]
	\label{ass:fixed_structure}
	$\exists\,(\theta,s)\in\Theta\times\mathcal{S}\ 
	\text{s.t.}\ \forall t\ge0,\;
	(\theta_t,s_t)=(\theta,s).$
\end{assum}

Under Assumption~\ref{ass:fixed_structure},
\eqref{eq:filter_map_s} defines the baseline
fixed-structure belief dynamics
(cf.\ Definition~\ref{def:fixed_structure_dynamics})
on $\mathcal{P}(\mathcal{Z})$.
This assumption establishes the fixed-structure
baseline against which CF adaptation is measured;
it is relaxed by the structural selection rule
introduced below.

\begin{rem}[Nonlinearity of belief dynamics]
	The filtering operator $\mathcal{F}_{\theta,s}:\mathcal{P}(\mathcal{Z}) \times \mathcal{U}\times \mathcal{Y}\to \mathcal{P}(\mathcal{Z})$
	is nonlinear in its belief argument. In particular, for
	$\mathfrak{B}_1,\mathfrak{B}_2\in\mathcal{P}(\mathcal{Z})$ and
	$\alpha\in[0,1]$,
	$
	\mathcal{F}_{\theta,s}(\alpha \mathfrak{B}_1 + (1-\alpha)\mathfrak{B}_2,u_t,y_{t+1})
	\neq
	\alpha \mathcal{F}_{\theta,s}(\mathfrak{B}_1,u_t,y_{t+1})
	+
	(1-\alpha)\mathcal{F}_{\theta,s}(\mathfrak{B}_2,u_t,y_{t+1}).
	$
	Equivalently, $\mathcal{F}_{\theta,s}$ is not affine on $\mathcal{P}(\mathcal{Z})$, i.e.,
	$
	\mathcal{F}_{\theta,s}\notin \mathrm{Aff}\big(\mathcal{P}(\mathcal{Z})\big).
	$
\end{rem}	

	For each $(\theta,s)\in\Theta\times\mathcal{S}$, define the prediction operator

\begin{equation}
	\mathcal{P}_{\theta,s}(\mathfrak{B}_t,u_t)
	:=
	\int
	\underbrace{p_{\theta,s}(z^+\mid z,u_t)}_{\text{state transition density}}
	\;
	\underbrace{\mathfrak{B}_t(dz)}_{\text{current belief}}
	,
	\label{eq:innovation}
\end{equation}
which yields the one-step predictive belief under the transition model
specified by structure $s$.

The consistency of the predicted belief with an incoming observation $y_{t+1}$
is quantified by the innovation likelihood
\begin{equation}
	\ell_{\theta,s}(y_{t+1} \mid \mathfrak{B}_t, u_t)
	:=
	\int
	p_{\theta,s}(y_{t+1} \mid z)\,
	\underbrace{(\mathcal{P}_{\theta,s}\mathfrak{B}_t, u_t)}_{\text{prediction}}(dz),
	\label{eq:innovation_likelihood}
\end{equation}
which is the marginal likelihood of $y_{t+1}$ under 
$\mathcal{P}_{\theta,s}(\mathfrak{B}_t,u_t)$.

Under standard regularity conditions, the Bayesian correction step \cite{jazwinski1970stochastic,maybeck1979stochastic} is given by
\begin{equation}
	\mathfrak{B}_{t+1}(dz)
	=
	\frac{
		p_{\theta,s}(y_{t+1} \mid z)\,
		\mathcal{P}_{\theta,s}(\mathfrak{B}_t, u_t)(dz)
	}{
		\ell_{\theta,s}(y_{t+1} \mid \mathfrak{B}_t, u_t)
	}, 
	\label{eq:bayes_update}
\end{equation}
which, together with \eqref{eq:innovation_likelihood}, defines a nonlinear,
input-driven update
$
(\mathfrak{B}_t,u_t,y_{t+1})\;\mapsto\;\mathfrak{B}_t^+
\in \mathcal{P}(\mathcal{Z}).
$

For fixed $(\theta,s)$, this update fully determines the belief evolution from
$(\mathfrak{B}_t,u_t,y_{t+1})$.
Accordingly, we define the structural inconsistency score by 
\begin{equation}
	\Phi(\mathfrak{B}_t, s)
	:=
	- \log \ell_{\theta,s}(y_{t+1} \mid \mathfrak{B}_t, u_t),
	\label{eq:Phi_def}
\end{equation}
so that smaller values of $\Phi$ indicate better predictive alignment.


Crucially, $\mathfrak{B}_{t+1}$
in~\eqref{eq:bayes_update} may remain well posed
$\forall{t}$, i.e.,
$\mathfrak{B}_{t+1}=\mathcal{F}_{\theta,s}(\mathfrak{B}_t,u_t,y_{t+1})$
in~\eqref{eq:filter_map_s} is computable at each step, while the resulting
belief sequence $\{\mathfrak{B}_t\}$ fails to converge to
$\mathbb{P}^\star(\cdot\mid y_{1:t},u_{1:t-1})$.

\begin{defn}[Structural mismatch]
	\label{def:struct_mismatch}
	We call $s\in\mathcal{S}$ structurally mismatched if
	$
	\exists\,\varepsilon>0 \ \text{s.t.}\ 
	\forall\,\{\theta_t\}\subset\Theta,\ 
	\liminf_{t\to\infty}
	D_{\mathcal{KL}}\!\Big(
	\mathbb{P}^\star(\cdot\mid y_{1:t},u_{1:t-1})
	\,\Big\|\,\mathfrak{B}_t^{\theta_t,s}
	\Big)\ge\varepsilon.
	$
\end{defn}

Thus, adaptation within the fixed structure $s$ cannot eliminate the
asymptotic discrepancy with the true conditional law.

When $s_t$ is structurally mismatched in the sense of
Definition~\ref{def:struct_mismatch}, 
the structural update $s_{t+1}$\footnote{%
	The variable $s_{t+1}$ denotes the \emph{selected latent structure} at time $t+1$.
	It is a discrete structural index chosen deterministically from the finite set
	$\mathcal{S}$ based on the current belief $\mathfrak{B}_t$.
	It is not a random variable and is not part of the Bayesian state; rather, it
	indexes the observation/transition model under which the subsequent Bayesian
	belief update is performed.
} is given by

\begin{equation}
	s_{t+1} = \begin{cases}
		s_t, ~~~~~~~ \Phi(\mathfrak{B}_t, s_t) \leq \min_{s\in\mathcal{S}}\Phi(\mathfrak{B}_t,s) + \delta,\\
		\arg\min_{s\in\mathcal{S}} \Phi(\mathfrak{B}_t, s), ~~~~\text{otherwise,}
	\end{cases}
   \label{eq:cf_select}  
\end{equation}
where $\delta \geq 0$ is a hysteresis margin.
Setting $\delta = 0$ recovers the pure argmin rule.

The belief $\mathcal{B}_{t+1}$ in~\eqref{eq:joint_update} is then obtained via the structure-conditioned Bayesian filter in~\eqref{eq:filter_map_s}.

Algorithm~\ref{alg:cf_update} provides a constructive realization of
\eqref{eq:cf_select} and \eqref{eq:filter_map_s}.


	\begin{algorithm}[t]
	\footnotesize
	\caption{CF belief--structure update at time $t$:
		constructive realization of the CF selection rule~\eqref{eq:cf_select}
		and the coupled recursion~\eqref{eq:cf_struct_update}--\eqref{eq:cf_belief_update}}
	\label{alg:cf_update}
	\small
	\begin{algorithmic}[1]
		\Require Current belief--structure pair $(\mathfrak{B}_t,s_t)$, candidate set $\mathcal{S}$, input $u_t$, measurement $y_{t+1}$   \hfill \eqref{eq:plant_state}–\eqref{eq:plant_obs} and Definition~\ref{def:struct_mismatch}
		\Ensure Updated pair $(\mathfrak{B}_{t+1}, s_{t+1})$     \hfill  \eqref{eq:belief_space} and \eqref{eq:filter_map_s}
		
		\State \textbf{Score evaluation:}
		\ForAll{$s \in \mathcal{S}$}
		\State Compute $\Phi(\mathfrak{B}_t,s)$ \hfill \eqref{eq:Phi_def}
		\EndFor
		
		\State \textbf{Structure selection:}
		\State Select $s_{t+1} \in \arg\min_{s\in\mathcal{S}} \Phi(\mathfrak{B}_t,s)$
		according to \hfill \eqref{eq:cf_select}
		
		\State \textbf{Belief propagation:}
		\State $\mathfrak{B}_{t+1} \leftarrow
		\mathcal{F}_{\theta,s_{t+1}}(\mathfrak{B}_t,u_t,y_{t+1})$
		\hfill \eqref{eq:filter_map_s}
		
	\end{algorithmic}
\end{algorithm}

However, minimizers of \eqref{eq:cf_select} need not be unique, i.e.,
$\arg\min_{s \in \mathcal{S}} \Phi(\mathfrak{B}_t,s)$ is not a singleton.
Under structural mismatch (Definition~\ref{def:struct_mismatch}),
$\min_{s \in \mathcal{S}} \Phi(\mathfrak{B}_t,s) > 0$, and
$\exists\, s_1 \neq s_2 \in \mathcal{S}$ such that
$\Phi(\mathfrak{B}_t,s_1) = \Phi(\mathfrak{B}_t,s_2)
= \min_{s \in \mathcal{S}} \Phi(\mathfrak{B}_t,s)$,
so \eqref{eq:cf_select} admits multiple minimizers.

To obtain a well-defined recursion, we introduce a deterministic selection operator 
that resolves this ambiguity:
\begin{equation}
	\mathcal{T}_{\mathrm{CF}} :
	\mathcal{P}(\mathcal{Z}) \times \mathcal{S}
	\;\to\; \mathcal{S},
	\label{eq:TCF}
\end{equation}
which selects a unique element from the set of minimizers of \eqref{eq:cf_select}, 
i.e., $\mathcal{T}_{\mathrm{CF}}(\mathfrak{B}_t,s_t) \in 
\arg\min_{s \in \mathcal{S}} \Phi(\mathfrak{B}_t, s)$.

Accordingly, CF induces the coupled belief--structure recursion
\begin{align}
	s_{t+1} &= \mathcal{T}_{\mathrm{CF}}(\mathfrak{B}_t,s_t),
	\label{eq:cf_struct_update}\\
	\mathfrak{B}_{t+1} &=
	\mathcal{F}_{\theta,s_{t+1}}(\mathfrak{B}_t,u_t,y_{t+1}),
	\label{eq:cf_belief_update}
\end{align}
where $\mathcal{F}_{\theta,s}$ denotes the Bayesian filtering operator under
structure $s$ (cf.~\eqref{eq:filter_map_s}).
Together, \eqref{eq:cf_struct_update}--\eqref{eq:cf_belief_update}
define the closed-loop evolution of the CF-augmented inference system.



To formalize the requirement that CF mitigates persistent
structural inconsistency—such as that quantified by
Definition~\ref{def:struct_mismatch}—we introduce the following design
assumption.


\begin{assum}[Structural inconsistency functional]
	\label{ass:struct_improve}
	$	\exists\, \Phi:\mathcal{P}(\mathcal{Z}) \times \mathcal{S} \to \mathbb{R}_+ \ \text{s.t.} \
	\forall (\mathfrak{B},s),\ \Phi(\mathfrak{B},s) \ge 0,
	\quad
	\Phi(\mathfrak{B},s) = 0 \iff s \in \mathcal{S}^\star.$
\end{assum}


Practically, $\Phi$ can be constructed from predictive or innovation errors
evaluated under the model associated with $s$. 
Accordingly, the CF augmented inference mechanism induced by
\eqref{eq:cf_struct_update}--\eqref{eq:cf_belief_update} can be written as
$(\mathfrak{B}_t,s_t) \mapsto
(\mathcal{T}_{\mathrm{CF}}(\mathfrak{B}_t,s_t),
\mathcal{F}_{\theta,\mathcal{T}_{\mathrm{CF}}(\mathfrak{B}_t,s_t)}
(\mathfrak{B}_t,u_t,y_{t+1}))$.
In particular, \eqref{eq:cf_belief_update} remains Bayesian, whereas \emph{CF
	acts only through the structural update} \eqref{eq:cf_struct_update}.  
In a system-theoretic viewpoint, 
the operator
$\mathcal{T}_{\mathrm{CF}}$ in~\eqref{eq:cf_struct_update} enlarges the set of
admissible belief trajectories 
$
\mathfrak{B}_t \in \mathcal{R}(s)
:= \Big\{\mathcal{F}_{\theta,s}^{(t)}(\mathfrak{B}_0,u_{0:t-1},y_{1:t})
:\theta \in \Theta_s\Big\},
$
associated with 
$s \in \mathcal{S}$, to 
$
\mathfrak{B}_t \in \bigcup_{s_{0:t} \in \mathcal{S}^{t+1}}
\mathcal{R}(s_{0:t}),
$
where $\mathcal{R}(s_{0:t})$ denotes the set of belief trajectories generated
by the switching sequence $s_{0:t}$ under~\eqref{eq:cf_struct_update}--\eqref{eq:cf_belief_update}.
Thus, CF enables escape from regimes of structural mismatch, i.e.,
$
\inf_{\theta \in \Theta_s}
D_{\mathcal{KL}}\!\big(
\mathbb{P}^\star(\cdot\mid y_{1:t},u_{1:t-1})
\,\|\,\mathfrak{B}_t^{\theta,s}
\big) \;\ge\; \varepsilon > 0,  \forall s \in \mathcal{S}.
$

\begin{rem}[Constructive realization] 
	The innovation score~\eqref{eq:Phi_def}, namely 
	$\Phi(B_t,s) = -\log\ell_{\theta,s}(y_{t+1}|B_t,u_t)$,
	satisfies Assumption~\ref{ass:struct_improve} whenever the observation model
	$\{p_{\theta,s}(\cdot|\cdot)\}_{s\in\mathcal{S}}$ is
	identifiable, in the sense that
	$\ell_{\theta,s}(y|B,u) = \ell_{\theta,s^\star}(y|B,u)$
	a.s. implies $s = s^\star$.
	This follows from the strict positivity of the KL
	divergence: $D_\mathrm{KL}(p_{\theta,s^\star}\|p_{\theta,s}) > 0$
	for $s\neq s^\star$ under identifiability.
\end{rem}

\begin{prop}[Innovation $\&$ CF switching]
	\label{prop:cf_switch_consistency}
	Let Assumption~\ref{ass:struct_improve} hold. Define the innovation cost
	$
	c_t^{(s)} := -\log \ell_{\theta,s}(y_t \mid \mathfrak{B}_{t-1},u_{t-1}).
	$
	Assume that there exists $\delta > 0$ such that
	$
	\liminf_{t\to\infty}
	\Big(
	\frac{1}{t} \sum_{k=1}^t c_k^{(s)}
	-
	\inf_{s' \in \mathcal{S}} \frac{1}{t} \sum_{k=1}^t c_k^{(s')}
	\Big)
	\ge \delta,
	\quad \forall s \notin \mathcal{S}^\star.
	$
	Then, the CF selection rule~\eqref{eq:cf_struct_update} satisfies
	$
	\limsup_{t\to\infty} \mathbf{1}\{s_t \notin \mathcal{S}^\star\} = 0,
	$
	i.e., structural selections outside $\mathcal{S}^\star$ occur only finitely often.
\end{prop}

\begin{pf}
	For any $s \notin \mathcal{S}^\star$, the separation condition implies
	$\exists\, \delta>0$ and $T<\infty$ such that $\forall t \ge T$,
	$
	\frac{1}{t}\sum_{k=1}^t c_k^{(s)}
	\;\ge\;
	\inf_{s'\in\mathcal{S}} \frac{1}{t}\sum_{k=1}^t c_k^{(s')}
	+ \delta.
	$
	By asymptotic consistency of $\Phi(\mathfrak{B}_t,s)$ with $\{c_t^{(s)}\}$,
	this yields
	$
	s \notin \arg\min_{s\in\mathcal{S}} \Phi(\mathfrak{B}_t,s),
	\quad \forall t \ge T.
	$
	Since $s_t \in \arg\min_{s} \Phi(\mathfrak{B}_t,s)$, it follows that
	$
	\mathbf{1}\{s_t \notin \mathcal{S}^\star\} = 0
	\quad \forall t \ge T,
	$
	hence $\limsup_{t\to\infty} \mathbf{1}\{s_t \notin \mathcal{S}^\star\}=0$.
\end{pf}


	The preceding results motivate a three-layer organization of the analysis,
aligned with the conceptual architecture as follow. 

\emph{Layer~1 (well-posedness and fixed-structure limitations).}
We first establish well-posedness of the structure-conditioned Bayesian recursion
$\mathfrak{B}_{t+1}=\mathcal{F}_{\theta,s}(\mathfrak{B}_t,u_t,y_{t+1})$
on $\mathcal{P}(\mathcal{Z})$ (Lemma~\ref{lem:belief_invariance}).
We then show that the coupled recursion
$\big(s_{t+1},\mathfrak{B}_{t+1}\big)$ given by
\eqref{eq:cf_struct_update}--\eqref{eq:cf_belief_update}
is well posed on $\mathcal{P}(\mathcal{Z})\times\mathcal{S}$, in the sense of
a unique forward-invariant trajectory for any input--output sequence
(Theorem~\ref{thm:wellposedness_cf}).
Next, for structurally mismatched $s$ in the sense of
Definition~\ref{def:struct_mismatch}, we show that no
(possibly time-varying) $\theta_t$ can restore asymptotic predictive
consistency within that fixed $s$ (Theorem~\ref{thm:nonparametric_mismatch}).
Finally, we show that allowing $s_{t+1}\neq s_t$ enlarges the set of attainable
one-step belief updates relative to any fixed $s\in\mathcal{S}$
(Theorem~\ref{thm:reachable_set}).

\emph{Layer~2 (mechanism-level guarantees for CF).}
We analyze the structural update
$s_{t+1}=\mathcal{T}_{\mathrm{CF}}(\mathfrak{B}_t,s_t)$.
We first establish a one-step descent property of the score
$\Phi(\mathfrak{B}_t,s)$ under \eqref{eq:cf_struct_update}
(Lemma~\ref{lem:struct_descent}).
We then show that persistent separation of $\Phi(\mathfrak{B}_t,s)$ implies
finite switching and eventual absorption into a single structure
(Lemma~\ref{lem:finite_switching}).
The coupled recursion
$(s_{t+1},\mathfrak{B}_{t+1})$
is interpreted as a hybrid dynamical system on
$\mathcal{P}(\mathcal{Z})\times\mathcal{S}$
(Proposition~\ref{prop:hybrid}).
Combining these results yields bounded $\{\mathfrak{B}_t\}$ and monotone
(and, under mismatch, strict) improvement of predictive consistency
(Theorem~\ref{thm:cf_boundedness}).

\emph{Layer~3 (behavioral consequences).}
We characterize $(s_t,\mathfrak{B}_t)$ asymptotically.
If $s_t\to s^\star\in\mathcal{S}^\star$, then
$s_{t+1}=s_t$ and
$\mathfrak{B}_{t+1}=\mathcal{F}_{\theta,s^\star}(\mathfrak{B}_t,u_t,y_{t+1})$
(Corollary~\ref{cor:eventual_fixed_filtering}).
If $s\in\mathcal{S}^\star$, then
$\mathcal{T}_{\mathrm{CF}}(\mathfrak{B}_t,s)=s$ eventually, i.e.,
no persistent switching (Corollary~\ref{cor:non_intrusive}).


	\subsection{Well-posedness (foundational, necessary)}
\label{sec:layer1}
\begin{lem}[Invariance of the belief space]
	\label{lem:belief_invariance}
	Fix a latent structure $s\in\mathcal{S}$ and parameters $\theta\in\Theta$.
	For any input $u_t\in\mathcal{U}$ and observation $y_{t+1}\in\mathcal{Y}$ such that
	$\ell_{\theta,s}(y\mid \mathfrak{B}_t,u_t)>0$, the structure-conditioned filtering map
	$\mathcal{F}_{\theta,s}$ defined in~\eqref{eq:filter_map_s} satisfies
	$
	\mathcal{F}_{\theta,s}(\mathfrak{B}_t,u_t,y_{t+1})\in\mathcal{P}(\mathcal{Z}),
	\qquad \forall\,\mathfrak{B}_t\in\mathcal{P}(\mathcal{Z}).
	$
\end{lem}

\begin{pf}
	
	Fix $\mathfrak{B}\in\mathcal{P}(\mathcal{Z})$ and define
	$\mathfrak{B}_t^+ := \mathcal{F}_{\theta,s}(\mathfrak{B}_t,u_t,y_{t+1})$.
	By the Bayesian update~\eqref{eq:bayes_update}, $\mathfrak{B}_t^+$ is obtained
	by absolutely continuous reweighting of the prediction measure
	$\mathcal{P}_{\theta,s}(\mathfrak{B}_t,u_t)\in\mathcal{P}(\mathcal{Z})$
	with respect to the likelihood $p_{\theta,s}(y\mid z)$, followed by normalization
	via the innovation likelihood
	$\ell_{\theta,s}(y_t\mid \mathfrak{B}_t,u_t)$ defined in
	\eqref{eq:innovation_likelihood}.
	Since $p_{\theta,s}(y\mid z)\ge 0$, $\forall{z}\in\mathcal{Z}$ and
	$\ell_{\theta,s}(y\mid \mathfrak{B}_t,u_t)>0$, the resulting measure
	$\mathfrak{B}^+$ is nonnegative.
	Moreover, rewriting~\eqref{eq:innovation_likelihood} yields
	$
		\int_{\mathcal{Z}} \mathfrak{B}_t^+(dz)
		=
		\frac{1}{\underbrace{\ell_{\theta,s}(y\mid \mathfrak{B}_t,u_t)}_{\text{innovation likelihood}}}
		\int_{\mathcal{Z}}
		\underbrace{p_{\theta,s}(y\mid z)}_{\text{likelihood}}
		\;
		\underbrace{\mathcal{P}_{\theta,s}(\mathfrak{B}_t,u_t)(dz)}_{\text{prediction measure}}
		=
		1.
		$
	Hence $\mathfrak{B}_t^+$ is normalized and therefore belongs to
	$\mathcal{P}(\mathcal{Z})$.
	This establishes invariance of the belief space under the Bayesian
	filtering recursion, cf.~\cite{jazwinski1970stochastic,maybeck1979stochastic}.
\end{pf}

\begin{thm}[Well-posedness]
	\label{thm:wellposedness_cf}
	Suppose Assumptions~\ref{ass:fixed_structure}--\ref{ass:struct_improve} hold.
	Let $(\mathfrak{B}_0,s_0)\in\mathcal{P}(\mathcal{Z})\times\mathcal{S}$.
	Then, for any input--output sequence
	$\{(u_t,y_{t+1})\}_{t\ge0}$, the CF selection rule
	\eqref{eq:cf_select} together with the coupled recursion
	\eqref{eq:cf_struct_update}--\eqref{eq:cf_belief_update}
	generates a unique sequence
	$\{(\mathfrak{B}_t,s_t)\}_{t\ge0}$ satisfying
	$
	(\mathfrak{B}_t,s_t)\in\mathcal{P}(\mathcal{Z})\times\mathcal{S},
	\qquad \forall t\ge0.
	$
	Equivalently, the induced coupled CF dynamics define a
	causal discrete-time hybrid system that is well posed and forward
	invariant on the admissible domain
	$\mathcal{P}(\mathcal{Z})\times\mathcal{S}$.
\end{thm}

\begin{pf}  
	By Assumption~\ref{ass:struct_improve}, the structural score
	$\Phi(\mathfrak{B}_t,s)$ is well defined for every
	$(\mathfrak{B}_t,s)\in\mathcal{P}(\mathcal{Z})\times\mathcal{S}$.
	Since $\mathcal{S}$ is finite, the minimization problem
	$
	\min_{s\in\mathcal{S}} \Phi(\mathfrak{B}_t,s)
	$
	in~\eqref{eq:cf_select}
	attains at least one minimizer for every
	$\mathfrak{B}_t\in\mathcal{P}(\mathcal{Z})$.
	Hence
	$\arg\min_{s\in\mathcal{S}} \Phi(\mathfrak{B}_t,s)\neq\varnothing$, $\forall{t}\ge0$.
	Because ties\footnote{i.e., the minimization problem
		$\min_{s\in\mathcal{S}}\Phi(\mathfrak{B}_t,s)$
		admits multiple minimizers.}
	are resolved deterministically in
	\eqref{eq:cf_select}, the selected structure
	$s_{t+1}\in\mathcal{S}$ is uniquely determined.
	Therefore the structure update
	\eqref{eq:cf_struct_update}
	is well defined $\forall{t}\ge0$.
	Next, by Assumption~\ref{ass:fixed_structure}, for each
	$s\in\mathcal{S}$, the structure-conditioned filtering operator
	$\mathcal{F}_{\theta,s}:\mathcal{P}(\mathcal{Z})\to\mathcal{P}(\mathcal{Z})$
	is well defined.
	Hence, once $s_{t+1}$ is determined, the belief update
	\eqref{eq:cf_belief_update}
	yields a unique posterior
	$\mathfrak{B}_{t+1}\in\mathcal{P}(\mathcal{Z})$.
	Consequently, if
	$(\mathfrak{B}_t,s_t)\in\mathcal{P}(\mathcal{Z})\times\mathcal{S}$,
	then
	$(\mathfrak{B}_{t+1},s_{t+1})\in
	\mathcal{P}(\mathcal{Z})\times\mathcal{S}$.
	Thus the admissible domain
	$\mathcal{P}(\mathcal{Z})\times\mathcal{S}$
	is forward invariant under the coupled recursion
	\eqref{eq:cf_struct_update}--\eqref{eq:cf_belief_update}.
	The base case holds since
	$(\mathfrak{B}_0,s_0)\in\mathcal{P}(\mathcal{Z})\times\mathcal{S}$.
	An induction argument then establishes existence and uniqueness of
	the sequence $\{(\mathfrak{B}_t,s_t)\}_{t\ge0}$ $\forall{t}\ge0$.
	Finally, causality follows directly from
	\eqref{eq:cf_struct_update}--\eqref{eq:cf_belief_update},
	because $(\mathfrak{B}_{t+1},s_{t+1})$ depends only on
	$(\mathfrak{B}_t,s_t)$ and the current data $(u_t,y_{t+1})$.
	Hence the coupled belief--structure recursion is well posed.
\end{pf}

\begin{thm}[Structural mismatch irreducibility]
	\label{thm:nonparametric_mismatch}
	Let $s \in \mathcal{S}$ be fixed and
	$\{\theta_t\}_{t\ge 0}$ arbitrary.
	Let $\{\mathfrak{B}_t^{\theta_t,s}\}$ satisfy
	\eqref{eq:filter_map_s}.
	If $s$ is structurally mismatched
	(Definition~\ref{def:struct_mismatch}), then
	$
	\inf_{\{\theta_t\}}
	\liminf_{t\to\infty}
	D_{\mathcal{KL}}\!\Big(
	\mathbb{P}^\star(\cdot \mid y_{1:t}, u_{1:t-1})
	\,\big\|\,
	\mathfrak{B}_t^{\theta_t,s}
	\Big)
	> 0 .
	$
\end{thm}
\begin{pf}
	Fix an arbitrary parameter sequence $\{\theta_t\}_{t\ge 0}$ and let
	$\{\mathfrak{B}_t^{\theta_t,s}\}_{t\ge 0}$ be generated by
	\eqref{eq:filter_map_s} under the fixed structure $s$.
	Since $s$ is structurally mismatched in the sense of
	Definition~\ref{def:struct_mismatch}, $\exists\varepsilon>0$ such that,
	for every admissible parameter sequence $\{\theta_t\}_{t\ge 0}$,
	$
	\liminf_{t\to\infty}
	D_{\mathcal{KL}}\!\Big(
	\mathbb{P}^\star(\cdot \mid y_{1:t},u_{1:t-1})
	\,\big\|\,
	\mathfrak{B}_t^{\theta_t,s}
	\Big)
	\ge \varepsilon > 0 .
	$
	Hence the divergence
	$
	D_{\mathcal{KL}}\!\Big(
	\mathbb{P}^\star(\cdot \mid y_{1:t},u_{1:t-1})
	\,\big\|\,
	\mathfrak{B}_t^{\theta_t,s}
	\Big)
	$
	cannot converge to $0$ as $t\to\infty$.
	Therefore the belief sequence $\{\mathfrak{B}_t^{\theta_t,s}\}_{t\ge 0}$ is not
	asymptotically consistent with
	$\mathbb{P}^\star(\cdot \mid y_{1:t},u_{1:t-1})$.
	Hence no possibly time-varying parameter adaptation
	$\{\theta_t\}_{t\ge 0}$ can eliminate the discrepancy, which is intrinsic
	to the structural constraint imposed by $s$.
\end{pf}

The next result quantifies the representational benefit of structural
adaptation relative to fixed-structure filtering in adaptive identification theory
\cite{Ljung1999system,narendra1989stable}.



\begin{thm}[Admissible update expansion]
	\label{thm:reachable_set}
	Fix a parameter class $\Theta$ and consider the Bayesian filtering operator
	$\mathcal{F}_{\theta,s}$ defined in~\eqref{eq:filter_map_s}.
	For a fixed latent structure $s\in\mathcal{S}$, define the one-step reachable
	set of beliefs as
	$
	\mathcal{R}_s(\mathfrak{B},u,y)
	\;:=\;
	\bigl\{
	\mathcal{F}_{\theta,s}(\mathfrak{B},u,y)
	\;\big|\;
	\theta\in\Theta
	\bigr\}
	\;\subseteq\;
	\mathcal{P}(\mathcal{Z}).
	$
	Under CF, define the corresponding reachable set, in the sense of belief
	updates induced by uncertainty over admissible models
	(cf. reachable-set constructions
	\cite{Aswani2013}), as
	$
	\mathcal{R}_{\mathrm{CF}}(\mathfrak{B},u,y)
	\;:=\;
	\bigcup_{s\in\mathcal{S}}
	\mathcal{R}_s(\mathfrak{B},u,y).
	$
	Then, for any belief $\mathfrak{B}\in\mathcal{P}(\mathcal{Z})$, input $u$,
	observation $y$, and any fixed structure $s\in\mathcal{S}$,
	$
	\mathcal{R}_s(\mathfrak{B},u,y)
	\;\subseteq\;
	\mathcal{R}_{\mathrm{CF}}(\mathfrak{B},u,y).
	$
	Moreover, if $\exists{s_1}\neq s_2$ such that   
	$
	\mathcal{R}_{s_1}(\mathfrak{B},u,y)
	\;\neq\;
	\mathcal{R}_{s_2}(\mathfrak{B},u,y),
	\quad
	\exists\,(\mathfrak{B},u,y)\in
	\mathcal{P}(\mathcal{Z})\times\mathcal{U}\times\mathcal{Y}.
	$
	then the inclusion is strict for at least one $s\in\mathcal{S}$, i.e.,
	$
	\mathcal{R}_s(\mathfrak{B},u,y)
	\;\subsetneq\;
	\mathcal{R}_{\mathrm{CF}}(\mathfrak{B},u,y).
	$
\end{thm}

\begin{pf}
	By definition,
	$$
	\mathcal{R}_{\mathrm{CF}}(\mathfrak{B},u,y)
	=
	\bigcup_{s\in\mathcal{S}}
	\mathcal{R}_s(\mathfrak{B},u,y),
	$$
	and hence
	$
	\mathcal{R}_s(\mathfrak{B},u,y)
	\subseteq
	\mathcal{R}_{\mathrm{CF}}(\mathfrak{B},u,y),
	\quad
	\forall\, s\in\mathcal{S}.
	$
	If $\exists{(s_1,s_2)}\in\mathcal{S}$ such that  
	$$
	\mathcal{R}_{s_1}(\mathfrak{B},u,y)
	\;\neq\;
	\mathcal{R}_{s_2}(\mathfrak{B},u,y), 
	\exists\,(\mathfrak{B},u,y)\in
	\mathcal{P}(\mathcal{Z})\times\mathcal{U}\times\mathcal{Y},
	$$ then consequently
	$$\bigcup_{s\in\mathcal{S}} \mathcal{R}_s(\mathfrak{B},u,y)
	\supsetneq
	\mathcal{R}_{s_1}(\mathfrak{B},u,y)$$
	for at least one $s_1\in\mathcal{S}$, i.e., CF strictly enlarges the
	structure-conditioned reachable set. The claim follows.
\end{pf}


\begin{rem}[Representation-level reachability]
	Unlike classical reachable-set enlargements induced by parametric
	uncertainty or probabilistic hybrid dynamics~\cite{Aswani2013,Abate2008},
	CF enlarges admissible belief evolution through variation in the latent
	structure $s$, rather than through parameter variation within a fixed
	structure.
\end{rem}


\begin{rem}[Implication for observation shifts]
	Experiment~\ref{sec:exp2} illustrates a regime in which a change in the
	observation model $p_{\theta,s}(y\mid z) \mapsto \tilde p_{\theta,s}(y\mid z)$ destroys latent-state identifiability under any fixed
	structure, i.e.,
	$\mathfrak{B}_t \not\to \mathbb{P}^\star(\cdot \mid y_{1:t},u_{1:t-1})$. In that case, by Theorem~\ref{thm:reachable_set}, 
	CF preserves admissible belief evolution by
	switching across $s\in\mathcal{S}$ 
	rather than remaining confined to a
	single observation-induced belief manifold, $\mathcal{M}_s := \{\mathfrak{B} : \ell_{\theta,s}(y\mid \mathfrak{B},u)=\text{const}\}$.
\end{rem}

We next clarify how this enlargement differs fundamentally from probabilistic
	mode-mixing approaches such as IMM filtering~\cite{blom1988interacting,barshalom2001estimation}.

\begin{prop}[Reachable set expansion]
		\label{prop:belief_reachability}
		Let $\mathcal{S}$ be a finite set of latent structures and
		$\mathfrak{B}_0 \in \mathcal{P}(\mathcal{Z})$ an initial belief.
		For admissible input–output sequences
		$(u_t,y_t)_{t\ge 0}\in \mathcal{U}^{\mathbb{N}}\times\mathcal{Y}^{\mathbb{N}}$,
		define
		\begin{equation}
			\mathcal{B}_{\mathrm{IMM}}
			:=
			\Big\{ \{\mathfrak{B}_t\}_{t\ge0} \;\Big|\;
			\mathfrak{B}_{t+1} \in 
			\operatorname{co}\!\big(
			F_{\theta,s}(\mathfrak{B}_t,u_t,y_{t+1}),
			\ s\in\mathcal{S}
			\big)
			\Big\},
			\label{eq:imm_update}
		\end{equation} 
		and
		\begin{equation}
			\mathcal{B}_{\mathrm{CF}}
			:=
			\Big\{ \{\mathfrak{B}_t\}_{t\ge0} \;\Big|\;
			\mathfrak{B}_{t+1}
			=
			F_{\theta,s_{t+1}}(\mathfrak{B}_t,u_t,y_{t+1}),
			\ s_{t+1}\in\mathcal{S}
			\Big\}.
			\label{eq:CF_update_cf}
		\end{equation}
		Then, in general,
		$
		\mathcal{B}_{\mathrm{IMM}} \subsetneq \mathcal{B}_{\mathrm{CF}} .
		$
\end{prop}

\begin{pf}
	By Theorem~\ref{thm:reachable_set}, for any $(\mathfrak{B}_t,u_t,y_{t+1})$,
	the admissible one-step update under CF strictly contains that of any fixed
	$s\in\mathcal{S}$. 
	For fixed $s$, define the trajectory class   
	$\mathcal{B}_s := \{\mathfrak{B}_{t+1} = F_{\theta,s}(\mathfrak{B}_t,u_t,y_{t+1})\}$,
	and let $\mathcal{B}_{\mathrm{CF}}$ be induced by~\eqref{eq:CF_update_cf} with
	$s_{t+1}\in\mathcal{S}$.
	In IMM filtering~\cite{barshalom2001estimation,10638633}, one has
	$\mathfrak{B}_{t+1}
	\in
	\operatorname{co}
	\{F_{\theta,s}(\mathfrak{B}_t,u_t,y_{t+1}) : s\in\mathcal{S}\}$,
	which defines $\mathcal{B}_{\mathrm{IMM}}$ in~\eqref{eq:imm_update}.
	This set is forward invariant, i.e.,
	$\mathfrak{B}_t \in \mathcal{B}_{\mathrm{IMM}} \Rightarrow
	\mathfrak{B}_{t+1} \in \mathcal{B}_{\mathrm{IMM}}$, $\forall{t}$.
	CF generates updates of the form
	$\mathfrak{B}_{t+1} = F_{\theta,s_{t+1}}(\mathfrak{B}_t,u_t,y_{t+1})$,
	with $s_{t+1}\in\mathcal{S}$, which are not restricted to the convex hull above.
	Hence there exists a switching sequence $\{s_t\}_{t\ge0}$ such that
	$\{\mathfrak{B}_t\}_{t\ge0} \not\subset \mathcal{B}_{\mathrm{IMM}}$,
	while $\{\mathfrak{B}_t\}_{t\ge0} \subset \mathcal{B}_{\mathrm{CF}}$.
	$\mathcal{B}_{\mathrm{IMM}} \subsetneq \mathcal{B}_{\mathrm{CF}}$.
\end{pf}

\begin{rem}
	Equation~\eqref{eq:cf_belief_update} defines a 
	one-step update, whereas
	\eqref{eq:imm_update} and~\eqref{eq:CF_update_cf}
	collect the corresponding belief trajectories under
	IMM mixing and CF structure selection;
	Proposition~\ref{prop:belief_reachability} lifts
	the one-step enlargement of
	Theorem~\ref{thm:reachable_set} to trajectory level.
\end{rem}
	\subsection{Structural adaptation mechanism \textbf{(core theory)}}	
\label{sec:layer2}
\begin{lem}[Structural descent]
	\label{lem:struct_descent}
	Under Assumption~\ref{ass:struct_improve}, the structural update
	$s_{t+1}=\mathcal{T}_{\mathrm{CF}}(\mathfrak{B}_t,s_t)$ in~\eqref{eq:cf_struct_update}
	satisfies
	\begin{equation}
		\Phi(\mathfrak{B}_t,s_{t+1}) \le \Phi(\mathfrak{B}_t,s_t),
		\label{eq:lemma_descent}
	\end{equation}
	with strict inequality whenever $s_t$ is structurally mismatched.
\end{lem}

\begin{pf}
	
	Fix $t$ and $s_t \in \mathcal{S}$. By \eqref{eq:cf_struct_update},
	$s_{t+1} = \mathcal{T}_{\mathrm{CF}}(\mathfrak{B}_t,s_t) \in \mathcal{S}$.
	By Assumption~\ref{ass:struct_improve},
	$
	s_{t+1} \in \arg\min_{s\in\mathcal{S}} \Phi(\mathfrak{B}_t,s),
	$
	hence
	$
	\Phi(\mathfrak{B}_t,s_{t+1}) \le \Phi(\mathfrak{B}_t,s_t),
	$
	which establishes \eqref{eq:lemma_descent}.
	If $s_t$ is structurally mismatched, then by Definition~\ref{def:struct_mismatch},
	$\exists\,\varepsilon>0$ such that
	$
	\inf_{\theta\in\Theta_{s_t}}
	D_{\mathcal{KL}}\!\left(
	\mathbb{P}^\star(\cdot \mid y_{1:t},u_{1:t-1})
	\,\middle\|\,
	\mathfrak{B}_t^{\theta,s_t}
	\right)
	\ge \varepsilon.
	$
	Thus $\exists\,\bar{s}\in\mathcal{S}$ such that
	$
	\Phi(\mathfrak{B}_t,\bar{s})
	\le
	\Phi(\mathfrak{B}_t,s_t) - \varepsilon.
	$
	By Assumption~\ref{ass:struct_improve},
	$
	\Phi(\mathfrak{B}_t,s_{t+1})
	\le
	\Phi(\mathfrak{B}_t,\bar{s})
	<
	\Phi(\mathfrak{B}_t,s_t),
	$
	which yields the strict inequality in \eqref{eq:lemma_descent}.
\end{pf}

\begin{lem}[Finite switching]
	\label{lem:finite_switching}
	Suppose $\exists\ s^\star \in \mathcal{S}$ and constants
	$\Delta > 0$ and $T_0 \in \mathbb{N}$ such that
	$
	\Phi_t(s^\star) \;\le\; \Phi_t(s) - \Delta,
	\qquad
	\forall s \in \mathcal{S}\setminus\{s^\star\},
	\;\forall t \ge T_0 .
	$
	Then, under the CF selection rule~\eqref{eq:cf_select} with hysteresis,
	the structure sequence $\{s_t\}$ switches only finitely many times and
	satisfies $s_t = s^\star$ for all sufficiently large $t$.
\end{lem}

\begin{pf}
	Let $\delta \in (0,\Delta)$. By assumption, $\forall{t} \ge T_0$ and $\forall{s} \neq s^\star$,
	$
	\Phi_t(s^\star) \le \Phi_t(s) - \Delta < \Phi_t(s) - \delta.
	$
	Hence, if $s_t \neq s^\star$, $t \ge T_0$, the hysteresis condition
	in~\eqref{eq:cf_select} implies
	$
	s_{t+1} = s^\star.
	$  
	If $s_t = s^\star$, then $\forall{s} \neq s^\star$,
	$
	\Phi_t(s^\star) < \Phi_t(s) - \delta,
	$
	so no switch is triggered, i.e.,
	$
	s_{t+1} = s^\star.
	$
	Thus, $s_t = s^\star$, $\forall{t} \ge T_0 + 1$.
	Since the interval $\{0,\dots,T_0\}$ is finite, the number of switches is finite.
\end{pf}


The next result interprets the coupled recursion as a hybrid system on
$\mathcal{P}(\mathcal{Z})\times\mathcal{S}$.


\begin{prop}[Hybrid belief--structure dynamics]
	\label{prop:hybrid}
	The coupled recursion \eqref{eq:cf_struct_update}--\eqref{eq:cf_belief_update}
	defines a discrete-time hybrid system on $\mathcal{P}(\mathcal{Z}) \times \mathcal{S}$:
	$
	(\mathfrak{B}_t,s_t)
	\mapsto
	(\mathfrak{B}_{t+1},s_{t+1}),
	$
	with
	$
	s_{t+1} = \mathcal{T}_{\mathrm{CF}}(\mathfrak{B}_t,s_t),
	\quad
	\mathfrak{B}_{t+1}
	=
	\mathcal{F}_{\theta,s_{t+1}}(\mathfrak{B}_t,u_t,y_{t+1}).
	$
	For fixed $s \in \mathcal{S}$,
	$
	\mathfrak{B}_{t+1}
	=
	\mathcal{F}_{\theta,s}(\mathfrak{B}_t,u_t,y_{t+1}).
	$
\end{prop}

\begin{pf}
	From \eqref{eq:cf_struct_update}--\eqref{eq:cf_belief_update},
	$
	s_{t+1} = \mathcal{T}_{\mathrm{CF}}(\mathfrak{B}_t,s_t) \in \mathcal{S},
	$
	$
	\mathfrak{B}_{t+1}
	=
	\mathcal{F}_{\theta,s_{t+1}}(\mathfrak{B}_t,u_t,y_{t+1}).
	$
	By Lemma~\ref{lem:belief_invariance},
	$
	\mathfrak{B}_{t+1} \in \mathcal{P}(\mathcal{Z}).
	$
	Thus the map
	$
	(\mathfrak{B}_t,s_t)
	\mapsto
	(\mathfrak{B}_{t+1},s_{t+1})
	$
	is well defined on $\mathcal{P}(\mathcal{Z}) \times \mathcal{S}$.
	For fixed $s$, the update reduces to
	$
	\mathfrak{B}_{t+1}
	=
	\mathcal{F}_{\theta,s}(\mathfrak{B}_t,u_t,y_{t+1}),
	$
	while $s_{t+1}$ evolves via $\mathcal{T}_{\mathrm{CF}}$.
\end{pf}

\begin{thm}[Boundedness and descent under CF]
	\label{thm:cf_boundedness}
	Let $\{(\mathfrak{B}_t,s_t)\}_{t\ge0}$ be generated by
	\eqref{eq:cf_struct_update}--\eqref{eq:cf_belief_update} on
	$\mathcal{P}(\mathcal{Z}) \times \mathcal{S}$. Then, $\forall{t}\ge0$:
	(i) (\emph{Belief invariance})
	$\mathfrak{B}_t \in \mathcal{P}(\mathcal{Z})$.
	(ii) (\emph{Monotone structural improvement})
	$\Phi(\mathfrak{B}_{t+1}, s_{t+1})
	\le
	\Phi(\mathfrak{B}_t, s_t)$.
	(iii) (\emph{Strict descent under mismatch})
	If $s_t$ is structurally mismatched in the sense of
	Definition~\ref{def:struct_mismatch} and \eqref{eq:cf_select}, then
	$\Phi(\mathfrak{B}_{t+1}, s_{t+1})
	<
	\Phi(\mathfrak{B}_t, s_t)$.
\end{thm}

\begin{pf}
	From \eqref{eq:cf_struct_update}--\eqref{eq:cf_belief_update},
	\[
	(\mathfrak{B}_{t+1}, s_{t+1})
	=
	\big(
	\mathcal{F}_{\theta,s_{t+1}}(\mathfrak{B}_t,u_t,y_{t+1}),
	\mathcal{T}_{\mathrm{CF}}(\mathfrak{B}_t,s_t)
	\big),\quad t\ge0.
	\]
	\emph{(i) Belief invariance.}
	By Lemma~\ref{lem:belief_invariance},
	$\mathcal{F}_{\theta,s}(\mathcal{P}(\mathcal{Z})) \subseteq \mathcal{P}(\mathcal{Z})$ for all $s\in\mathcal{S}$.
	Hence $\mathfrak{B}_{t+1}\in\mathcal{P}(\mathcal{Z})$ whenever $\mathfrak{B}_t\in\mathcal{P}(\mathcal{Z})$,
	and thus $\mathfrak{B}_t\in\mathcal{P}(\mathcal{Z})$ for all $t\ge0$.	
	\emph{(ii) Monotone structural improvement.}
	From \eqref{eq:cf_struct_update}, 
	$s_{t+1}\in\arg\min_{s\in\mathcal{S}}\Phi(\mathfrak{B}_t,s)$.
	By Lemma~\ref{lem:struct_descent},
	$\Phi(\mathfrak{B}_t,s_{t+1}) \le \Phi(\mathfrak{B}_t,s_t)$ for all $t\ge0$.
	\emph{(iii) Strict descent under mismatch.}
	If $s_t$ is structurally mismatched, then by Assumption~\ref{ass:struct_improve},
	$\exists\,\tilde{s}\in\mathcal{S}$ such that 
	$\Phi(\mathfrak{B}_t,\tilde{s}) < \Phi(\mathfrak{B}_t,s_t)$.
	Since $s_{t+1}\in\arg\min_{s\in\mathcal{S}}\Phi(\mathfrak{B}_t,s)$,
	$
	\Phi(\mathfrak{B}_t,s_{t+1})
	\le \Phi(\mathfrak{B}_t,\tilde{s})
	< \Phi(\mathfrak{B}_t,s_t).
	$
\end{pf}

	\subsection{Behavioral consequence (core corollary)}
\label{sec:layer3}

\begin{cor}[Fixed-structure reduction]
	\label{cor:eventual_fixed_filtering}
	Suppose the CF selection rule \eqref{eq:cf_select} is implemented with a
	hysteresis margin $\delta>0$.
	If $\exists{s}^\star\in\mathcal{S}$, $\Delta>\delta$, and $T_0\in\mathbb{N}$
	such that
	\begin{equation}
		\Phi(\mathfrak{B}_t,s^\star)
		\;\le\;
		\Phi(\mathfrak{B}_t,s) - \Delta,
		\qquad
		\forall s\in\mathcal{S}\setminus\{s^\star\},\;\forall t\ge T_0,
		\label{eq:score_separation}
	\end{equation}
	then the structure sequence $\{s_t\}$ switches only finitely many times and,
	for all sufficiently large $t$, satisfies $s_t=s^\star$.
	Consequently, the coupled CF recursion
	\eqref{eq:cf_struct_update}--\eqref{eq:cf_belief_update} reduces after a finite
	transient to the fixed-structure Bayesian filter
	\begin{equation}
		\mathfrak{B}_{t+1}
		=
		\mathcal{F}_{\theta,s^\star}(\mathfrak{B}_t,u_t,y_{t+1}),
		\qquad t\ \text{sufficiently large}.
		\label{eq:eventual_fixed_recursion}
	\end{equation}
\end{cor}

\begin{pf}
	Let the hysteresis version of \eqref{eq:cf_select} be written explicitly as:
	for each $t\ge 0$,
\begin{equation}
		s_{t+1}=
		\begin{cases}
			s_t, & \text{if } \Phi(\mathfrak{B}_t,s_t)\le \min\limits_{s\in\mathcal{S}}\Phi(\mathfrak{B}_t,s)+\delta,\\[1mm]
			\arg\min\limits_{s\in\mathcal{S}}\Phi(\mathfrak{B}_t,s), & \text{otherwise},
		\end{cases}
		\label{eq:hyst_rule}
	\end{equation}
	with $\delta>0$.
	(Any equivalent ``switch only if improvement exceeds $\delta$'' rule yields the same conclusion.)
	Assume \eqref{eq:score_separation}. Fix any $t\ge T_0$. Then $\forall{s}\neq s^\star$,
	$
	\Phi(\mathfrak{B}_t,s^\star)\le \Phi(\mathfrak{B}_t,s)-\Delta
	\quad\Longrightarrow\quad
	\min_{s\in\mathcal{S}}\Phi(\mathfrak{B}_t,s)=\Phi(\mathfrak{B}_t,s^\star).
	$
	In particular, if $s_t=s^\star$, then
	$
	\Phi(\mathfrak{B}_t,s_t)-\min_{s\in\mathcal{S}}\Phi(\mathfrak{B}_t,s)
	=
	\Phi(\mathfrak{B}_t,s^\star)-\Phi(\mathfrak{B}_t,s^\star)=0\le \delta,
	$
	and therefore \eqref{eq:hyst_rule} gives $s_{t+1}=s_t=s^\star$.
	This shows that $s^\star$ is \emph{absorbing} after time $T_0$.
	It remains to show that $s^\star$ is reached in finite time.
	For any $t\ge T_0$ with $s_t\neq s^\star$, separation implies
	$
	\Phi(\mathfrak{B}_t,s^\star)\le \Phi(\mathfrak{B}_t,s_t)-\Delta
	\quad\Longrightarrow\quad
	\Phi(\mathfrak{B}_t,s_t)>\min_{s\in\mathcal{S}}\Phi(\mathfrak{B}_t,s)+\delta,
	$
	because $\Delta>\delta$ and $\min_s \Phi(\mathfrak{B}_t,s)=\Phi(\mathfrak{B}_t,s^\star)$.
	Hence the first case in \eqref{eq:hyst_rule} cannot occur; a switch is triggered and
	$
	s_{t+1}=\arg\min_{s\in\mathcal{S}}\Phi(\mathfrak{B}_t,s)=s^\star.
	$
	Thus, regardless of the pre-$T_0$ history, we obtain $s_{T_0+1}=s^\star$, and by
	absorption, $s_t=s^\star$, $\forall{t}\ge T_0+1$.
	In particular, the number of switches after $T_0$ is at most one, so the total number
	of switches is finite.
	
	Finally, substituting $s_t=s^\star$, $\forall{t}\ge T_0+1$, into the coupled update
	\eqref{eq:cf_belief_update} yields the fixed-structure recursion
	\eqref{eq:eventual_fixed_recursion}.
\end{pf}


\begin{rem}[Connection to Experiment~\ref{sec:exp3}]
	\label{rem:exp3_connection}
	Experiment~\ref{sec:exp3} (negative control) is designed so that the true observation
	mechanism remains consistent with $s_{\mathrm{lin}}$; empirically,
	$\Phi(\mathfrak{B}_t,s_{\mathrm{lin}})$ remains persistently lower than
	$\Phi(\mathfrak{B}_t,s_{\mathrm{sat}})$, so CF rapidly settles on $s_{\mathrm{lin}}$
	and behaves as a standard fixed-\texttt{LIN} Bayesian filter thereafter. 
\end{rem}


\begin{cor}[Non-intrusiveness]
	\label{cor:non_intrusive}
	Suppose $\exists{s}^\star \in \mathcal{S}$ such that
	\begin{equation}
		\Phi(\mathfrak{B}_t,s^\star)
		\;\le\;
		\Phi(\mathfrak{B}_t,s),
		\qquad
		\forall s \in \mathcal{S},\ \forall t \ge 1 .
		\label{eq:non_intrusive_condition}
	\end{equation}
	Then the CF mechanism is non-intrusive, in the sense
	that $s_t = s^\star$, $\forall{t}\ge 1$, and the coupled update
	\eqref{eq:cf_struct_update}--\eqref{eq:cf_belief_update} reduces to standard
	Bayesian filtering under the fixed structure $s^\star$.
\end{cor}


\begin{pf}
	Condition~\eqref{eq:non_intrusive_condition}  makes $s^\star$ a global minimizer of $\Phi(\mathfrak{B}_t,s)$, 
	$\forall{t}$. Hence \eqref{eq:cf_select} gives $s_{t+1}=s^\star$, $\forall{t}$,
	and \eqref{eq:cf_belief_update} reduces to the fixed-structure recursion 
	\eqref{eq:eventual_fixed_recursion}.
\end{pf}

\section{Numerical Experiments}
\label{sec:experiments}

Four experiments evaluate CF across complementary 
mismatch scenarios: structural mismatch in the latent 
dynamics (Experiment~\ref{sec:exp1}), an abrupt observation-model 
shift (Experiment~\ref{sec:exp2}), a negative control with no shift 
(Experiment~\ref{sec:exp3}), and a two-dimensional latent state 
(Experiment~\ref{sec:exp44}). Together, they test the three properties 
established in Section~\ref{sec:CFDeeepSSSM}: accuracy 
under mismatch, correctness of structural adaptation, and 
non-intrusiveness under correct specification.

Three metrics are reported throughout. State-estimation 
accuracy is measured by
\begin{align}
	\label{eq:RMSE}
	\mathrm{RMSE}
	:=
	\left(
	\frac{1}{T}\sum_{t=1}^{T}\|\hat{z}_t - z_t\|^2
	\right)^{1/2},
\end{align}
where $\hat{z}_t := \mathbb{E}_{\mathfrak{B}_t}[z_t]$.
Predictive consistency is quantified by the 
time-averaged innovation score
\begin{align}
	\label{expall:Mean}
	\bar{\Phi}
	:=
	\frac{1}{T-1}\sum_{t=1}^{T-1}\Phi_t(s_t),
\end{align}
and structural adaptation by the switch rate
\begin{equation}
	\label{eq:switch_rate}
	\rho_{\mathrm{sw}}
	:=
	\frac{1}{T-1}\sum_{t=1}^{T-1}
	\mathbf{1}\{s_{t+1}\neq s_t\}.
\end{equation}
All metrics are averaged over $M=50$ Monte~Carlo runs.


\begin{figure*}[t]
	\centering
	\includegraphics[width=1.16\textwidth]{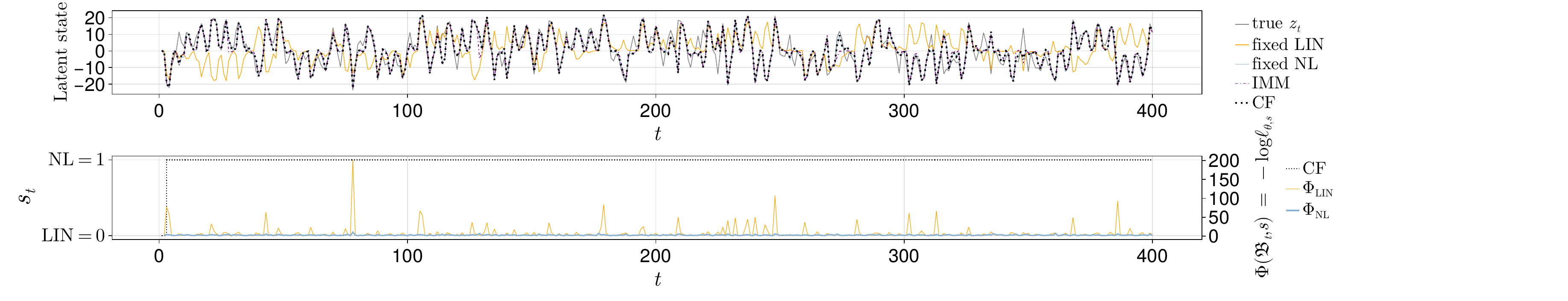} 
\caption{%
	\textbf{Experiment~\ref{sec:exp1}.}
	\textit{Top:} True state $z_t$ and estimates from
	\texttt{Fixed~LIN}, \texttt{Fixed~NL}, IMM, and CF,
	initialised at $s_0=s_{\mathrm{lin}}$.
	\textit{Bottom:} Structure sequence $s_t$
	(LIN$=0$, NL$=1$) and innovation scores
	$\Phi_{\mathrm{LIN}}$, $\Phi_{\mathrm{NL}}$;
	CF commits to $s_{\mathrm{nl}}$ at $t\approx 5$
	and produces no further switches.
}
	\label{fig:exp4A_benchmark}
\end{figure*}

\begin{figure*}[t]
	\centering
	\includegraphics[width=1.15\textwidth]{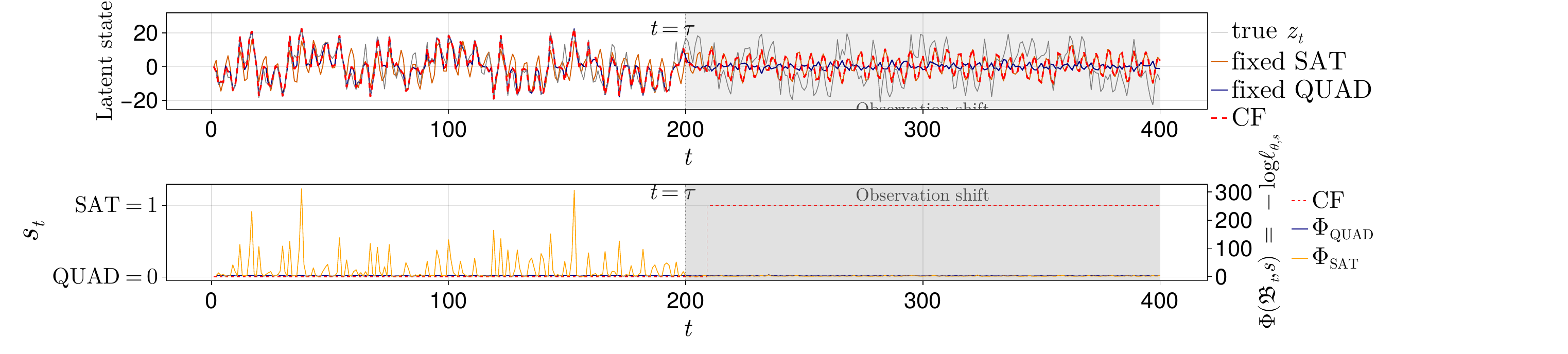} 
	\caption{\textbf{Experiment~\ref{sec:exp2}:}
		Top: True state $z_t$ and estimates from fixed QUAD, fixed SAT, and CF. 
		The dashed line marks the change time $\tau$. 
		Bottom: Selected structure $s_t$ (QUAD$=0$, SAT$=1$) and scores $\Phi_{\mathrm{QUAD}}$, $\Phi_{\mathrm{SAT}}$.}
	\label{fig:exp4B_benchmark}
\end{figure*}

\begin{figure*}[t]
	\centering
	\includegraphics[width=1.15\textwidth]{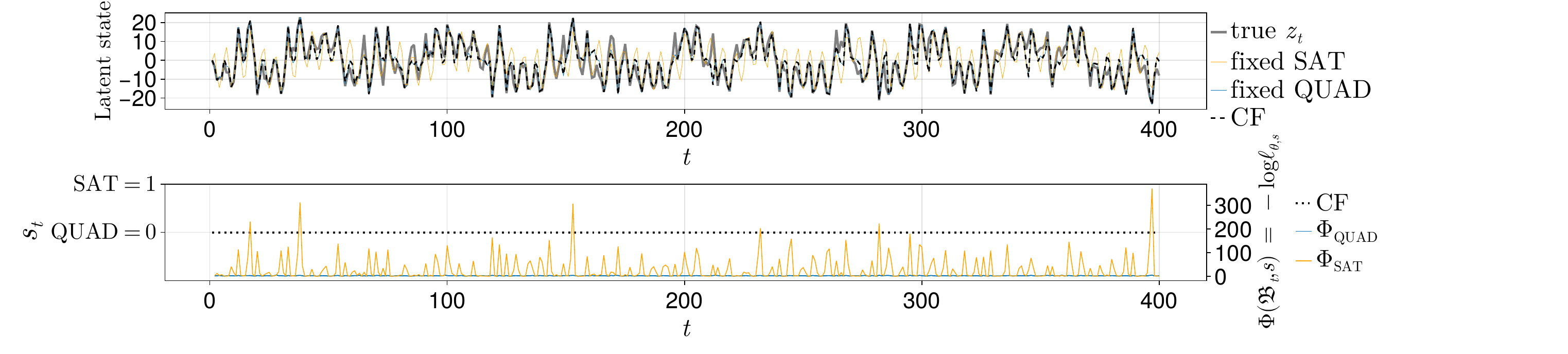} 
\caption{%
	\textbf{Experiment~\ref{sec:exp3}}
	(negative control, no observation shift).
	\textit{Top:} True state $z_t$ and estimates from
	\texttt{Fixed-QUAD}, \texttt{Fixed-SAT}, and CF (proposed);
	CF overlaps with \texttt{Fixed-QUAD} throughout.
	\textit{Bottom:} Structure sequence $s_t$
	(QUAD$=0$, SAT$=1$) and innovation scores
	$\Phi_{\mathrm{QUAD}}$, $\Phi_{\mathrm{SAT}}$;
	the selected structure remains at QUAD for all $t$
	and $\Phi_{\mathrm{QUAD}}$ stays below
	$\Phi_{\mathrm{SAT}}$ throughout.
}
\label{fig:exp4C_benchmark}
\end{figure*}


\subsection{Experiment~4.1: Structural mismatch in latent dynamics}
\label{sec:exp1}

The data-generating process is the canonical nonlinear 
stochastic growth model~\cite{Arulampalam2002PF},
widely used as a benchmark for nonlinear filtering
methods~\cite{2012ParticleFiltering,Kitagawa1996MonteCarlo,1992M}.
The scalar latent state $z_t \in \mathbb{R}$ evolves as
\begin{align}
	z_{t+1} &= \tfrac{1}{2}\,z_t
	+ \frac{25\,z_t}{1+z_t^2}
	+ 8\cos(1.2\,t) + w_t,
	\label{eq:exp4A_true_dyn}\\
	y_t     &= \tfrac{z_t^2}{20} + v_t,
	\label{eq:exp4A_obs}
\end{align}
with $w_t\sim\mathcal{N}(0,\sigma_w^2)$,
$v_t\sim\mathcal{N}(0,\sigma_v^2)$, and
$z_0\sim\mathcal{N}(0,\sigma_0^2)$.

\smallskip
\noindent\textbf{Candidate structures.}
Two competing transition hypotheses are considered:
$\mathcal{S}:=\{s_{\mathrm{lin}},\,s_{\mathrm{nl}}\}$.
Under $s_{\mathrm{lin}}$, the transition follows a 
linear--Gaussian model
$z_{t+1}\sim\mathcal{N}(\alpha z_t,\hat{\sigma}_w^2)$,
which cannot represent the nonlinear 
dynamics~\eqref{eq:exp4A_true_dyn} and thus induces 
structural mismatch.
Under $s_{\mathrm{nl}}$, the transition matches the 
true process,
$z_{t+1}\sim\mathcal{N}\!\bigl(
\tfrac{1}{2}z_t+\tfrac{25z_t}{1+z_t^2}
+8\cos(1.2t),\,\hat{\sigma}_w^2\bigr)$.
Both structures share the quadratic observation 
model~\eqref{eq:exp4A_obs}, so mismatch is isolated 
to the latent dynamics.

\smallskip
\noindent\textbf{Implementation.}
Each structure-conditioned belief is propagated via a
bootstrap particle filter with $N_p=2500$ particles.
The CF selection rule~\eqref{eq:cf_select} is applied
to the $W=10$-step windowed average
$\bar{\Phi}_t^W(s):=W^{-1}\sum_{k=0}^{W-1}\Phi_{t-k}(s)$
with hysteresis margin $\delta=1.0$, consistent with
Corollary~\ref{cor:eventual_fixed_filtering}.
The experiment is initialised at $s_0=s_{\mathrm{lin}}$
to test structural recovery from an incorrect starting
point.
Three methods are compared: \texttt{Fixed~LIN},
\texttt{Fixed~NL}, and the IMM
filter~\cite{blom1988interacting} ($p_{ii}=0.95$).

\smallskip
\noindent\textbf{Results.}
Figure~\ref{fig:exp4A_benchmark} shows that CF recovers 
the accuracy of \texttt{Fixed~NL} after a single structural 
transition at $t\approx5$, producing no further switches 
over the remaining horizon ($\rho_{\mathrm{sw}}=0.011$).
The innovation scores in the bottom panel make the 
mechanism transparent: $\Phi_{\mathrm{nl}}$ falls 
persistently below $\Phi_{\mathrm{lin}}$ after the 
initial transient, so the hysteresis condition is met 
exactly once and the structure locks to $s_{\mathrm{nl}}$.
This confirms the one-step structural descent property
(Lemma~\ref{lem:struct_descent}) and the finite-switching 
guarantee (Corollary~\ref{cor:eventual_fixed_filtering}).

IMM achieves accuracy comparable to \texttt{Fixed~NL}, 
but does so through probabilistic model mixing rather 
than a hard structural commitment. CF, by contrast, 
identifies and commits to the correct structure after 
a short transient, illustrating the distinction
$\mathcal{B}_{\mathrm{IMM}}\subsetneq\mathcal{B}_{\mathrm{CF}}$
(Proposition~\ref{prop:belief_reachability}).
Quantitative results are summarised in 
Table~\ref{tab:all_experiments}.

\begin{table}[t]
	\caption{%
		Performance across all experiments
		($M=50$ Monte~Carlo runs, $N_p=2500$ particles, 
		$T=400$; Exp.~4.4: $N_p=1000$, $T=200$, $M=100$).
		$\dagger$~IMM: self-transition probability $p_{ii}=0.95$.
		Lower is better for all metrics.
	}
	\label{tab:all_experiments}
	\centering
	\renewcommand{\arraystretch}{1.10}
	\setlength{\tabcolsep}{4pt}
	\scriptsize
	\begin{tabular}{llccc}
		\hline\hline
		Exp. & Method
		& RMSE $\downarrow$
		& $\bar{\Phi}$ $\downarrow$
		& $\rho_{\mathrm{sw}}$ \\
		\hline
		\multirow{4}{*}{4.1}
		& Fixed LIN              & 13.463 & 7.947  & --    \\
		& Fixed NL               & 10.273 & 4.192  & --    \\
		& IMM$^\dagger$          & 10.534 & --     & --    \\
		& \textbf{CF (ours)}     & \textbf{10.688}
		& \textbf{4.168}
		& \textbf{0.011} \\
		\hline
		\multirow{3}{*}{4.2}
		& Fixed-QUAD             & 8.415  & --     & --    \\
		& Fixed-SAT              & 8.291  & --     & --    \\
		& \textbf{CF (ours)}     & \textbf{7.408}
		& \textbf{2.071}
		& \textbf{0.003} \\
		\hline
		\multirow{3}{*}{4.3}
		& Fixed-QUAD             & 4.412  & --     & --    \\
		& Fixed-SAT              & 7.230  & --     & --    \\
		& \textbf{CF (ours)}     & \textbf{4.413}
		& \textbf{2.599}
		& \textbf{0.000} \\
		\hline
		\multirow{3}{*}{4.4}
		& Fixed LIN              & 18.665 & 22.479 & --    \\
		& Fixed NL               &  7.105 &  5.702 & --    \\
		& \textbf{CF (ours)}     & \textbf{7.157}
		& \textbf{5.763}
		& \textbf{0.005} \\
		\hline\hline
	\end{tabular}
\end{table}


\subsection{Experiment~4.2: Abrupt observation-model shift}
\label{sec:exp2}

This experiment tests whether CF detects and adapts to 
an abrupt change in the observation structure at an 
unknown time $\tau$, while the latent 
dynamics~\eqref{eq:exp4A_true_dyn} remain fixed 
throughout.

\smallskip
\noindent\textbf{Candidate structures.}
Two candidate observation models are 
considered: $\mathcal{S}:=\{s_{\mathrm{quad}},
\,s_{\mathrm{sat}}\}$,
where $s_{\mathrm{quad}}$ denotes the quadratic and 
$s_{\mathrm{sat}}$ the saturating observation structure.
Under $s_{\mathrm{quad}}$,
\begin{equation}
	y_t \sim \mathcal{N}\!\left(
	\frac{z_t^2}{20},\,\hat{\sigma}_v^2\right),
	\label{eq:exp2_model_obs_quad}
\end{equation}
while under $s_{\mathrm{sat}}$,
\begin{equation}
	y_t \sim \mathcal{N}\!\left(
	\tanh\!\left(\frac{z_t^2}{20}\right),\,
	\hat{\sigma}_v^2\right).
	\label{eq:exp2_model_obs_sat_benchmark}
\end{equation}
The true observation process 
follows~\eqref{eq:exp2_model_obs_quad} for $t<\tau$ 
and switches to~\eqref{eq:exp2_model_obs_sat_benchmark} 
at $\tau=200$. Both candidate structures use the same 
latent dynamics~\eqref{eq:exp4A_true_dyn}, so mismatch 
is isolated to the observation model.
Three methods are compared: \texttt{Fixed-QUAD}
($s_t\equiv s_{\mathrm{quad}}$),
\texttt{Fixed-SAT} ($s_t\equiv s_{\mathrm{sat}}$), 
and CF.

\smallskip
\noindent\textbf{Implementation.}
Score evaluations use $N_p=2000$ particles.
The CF selection rule~\eqref{eq:cf_select} is applied 
to $W=10$-step windowed average scores with hysteresis 
margin $\delta=1.0$, consistent with
Corollary~\ref{cor:eventual_fixed_filtering}     

\smallskip
\noindent\textbf{Results.}
Figure~\ref{fig:exp4B_benchmark} illustrates the 
two-phase behaviour induced by the observation shift.

\emph{Before $t=\tau$:}
The score ordering $\Phi_{\mathrm{quad}} < 
\Phi_{\mathrm{sat}}$ is maintained throughout, so the 
hysteresis condition is never triggered and CF produces 
zero spurious switches.
The CF estimate closely tracks \texttt{Fixed-QUAD} and 
the true state $z_t$, consistent with the 
non-intrusiveness guarantee
(Corollary~\ref{cor:non_intrusive}): when the active  
structure is already predictively consistent, CF reduces 
exactly to the corresponding fixed-structure filter.

\emph{After $t=\tau$:}
The shift to~\eqref{eq:exp2_model_obs_sat_benchmark} 
immediately reverses the score ordering — 
$\Phi_{\mathrm{quad}}$ rises sharply while 
$\Phi_{\mathrm{sat}}$ falls — and CF responds with a 
single structural transition to $s_{\mathrm{sat}}$, 
committing to it for all remaining steps 
($\rho_{\mathrm{sw}}=0.0025$).
This confirms finite switching
(Corollary~\ref{cor:eventual_fixed_filtering}) and 
the one-step structural descent property
(Lemma~\ref{lem:struct_descent}).
\texttt{Fixed-QUAD}, by contrast, becomes persistently 
biased because it continues to use the mismatched 
model~\eqref{eq:exp2_model_obs_quad}.

It is worth noting that neither CF nor \texttt{Fixed-SAT} 
fully recovers the large-amplitude variations of the true 
state after the shift. This is not a limitation of CF 
itself, but an inherent consequence of the saturating 
map~\eqref{eq:exp2_model_obs_sat_benchmark} being 
many-to-one: the latent state is not globally 
identifiable from observations after $t=\tau$. CF 
converges to the best predictively consistent model 
available, as guaranteed by 
Theorem~\ref{thm:cf_boundedness}.
Quantitative results are reported in 
Table~\ref{tab:all_experiments}.


\subsection{Experiment~4.3: No observation shift 
	(negative control)}
\label{sec:exp3}

This experiment asks whether CF remains non-intrusive 
when no structural change occurs — that is, when the 
active structure is already predictively consistent 
throughout the horizon. It uses the same latent 
dynamics~\eqref{eq:exp4A_true_dyn} and candidate 
observation structures as Experiment~\ref{sec:exp2}, 
but the true observation process coincides with the 
quadratic model~\eqref{eq:exp2_model_obs_quad} for 
all $t$: no shift occurs at any time.

\smallskip
\noindent\textbf{Implementation.}
The CF mechanism uses the same $W=10$-step windowed 
scores and hysteresis margin $\delta=1.0$ as in 
Experiments~\ref{sec:exp1} and~\ref{sec:exp2}, 
with no additional penalty or persistence counter.
This ensures that any difference in behaviour relative 
to Experiment~\ref{sec:exp2} is attributable solely 
to the absence of a shift, not to a change in 
hyperparameters.

\smallskip
\noindent\textbf{Results.}
Figure~\ref{fig:exp4C_benchmark} confirms that CF 
produces zero structural switches throughout the 
horizon ($\rho_{\mathrm{sw}}=0.000$). The score 
ordering $\Phi_{\mathrm{quad}} < \Phi_{\mathrm{sat}}$ 
is maintained at every step, so the hysteresis 
condition is never triggered. The CF estimate overlaps 
with \texttt{Fixed-QUAD} throughout, and both track 
the true state $z_t$ accurately.

This outcome directly validates 
Corollary~\ref{cor:non_intrusive} when the active 
structure is predictively consistent, CF introduces no 
overhead and reduces exactly to the corresponding 
fixed-structure Bayesian filter. Taken together with 
Experiment~\ref{sec:exp2}, these two experiments form 
a controlled pair — same hyperparameters, same 
candidate structures, same latent dynamics — that 
isolates the effect of the observation shift on CF 
behaviour. The contrast is sharp: a single shift at 
$\tau=200$ is sufficient to trigger exactly one 
structural transition in Experiment~\ref{sec:exp2}, 
while the absence of any shift here produces none.
Quantitative results are reported in 
Table~\ref{tab:all_experiments}.


\begin{figure*}[t]
        \centering
        \includegraphics[width=\textwidth]{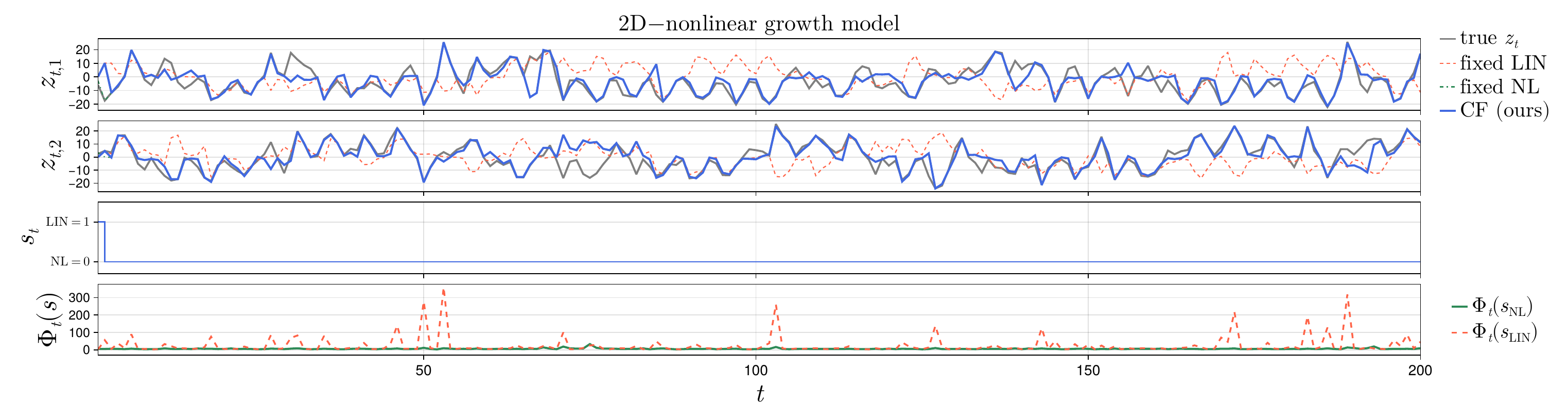}
\caption{%
        \textbf{Experiment~4.4}
        ($\mathcal{Z}=\mathbb{R}^{2}$, $T=200$,
        $N=1000$ particles, $s_0=s_{\mathrm{lin}}$).
        \textit{Panels 1--2:} Latent state estimates
        $z_{t,1}$ and $z_{t,2}$; CF recovers the accuracy of
        \texttt{Fixed~NL} in both dimensions after a single
        structural switch, while \texttt{Fixed~LIN} remains
        persistently biased.
        \textit{Panel 3:} Structure sequence $s_t$; CF commits
        to $s_{\mathrm{nl}}$ at $t\approx 2$ and produces no
        further switches.
        \textit{Panel 4:} Innovation scores
        $\Phi_t(s_{\mathrm{nl}})$ and $\Phi_t(s_{\mathrm{lin}})$;
        $\Phi_t(s_{\mathrm{nl}})$ remains persistently lower
        after $t\approx 2$.
}
        \label{fig:exp44}
\end{figure*}

\subsection{Experiment~4.4: Multidimensional latent state
	($\mathcal{Z}=\mathbb{R}^{2}$)}
\label{sec:exp44}

The theoretical results of 
Section~\ref{sec:CFDeeepSSSM} are stated for a general 
Polish space $\mathcal{Z}$ and do not rely on the latent 
state being scalar. This experiment confirms that the CF 
mechanism and its guarantees extend naturally to a 
two-dimensional setting, where higher score variance 
makes structure selection more challenging.

\smallskip
\noindent\textbf{System.}
The data-generating process extends the 
benchmark~\eqref{eq:exp4A_true_dyn}--\eqref{eq:exp4A_obs} 
to two independent dimensions with a phase 
offset $\varphi_{1}=0$ and $\varphi_{2}=1.0$\,rad:
\begin{align}
	z_{t+1,i} &= \tfrac{1}{2}\,z_{t,i}
	+ \frac{25\,z_{t,i}}{1+z_{t,i}^{2}}
	+ 8\cos(1.2\,t + \varphi_{i})
	+ w_{t,i},
	\label{eq:dyn2d}\\
	y_{t,i}   &= \tfrac{1}{20}\,z_{t,i}^{2} + v_{t,i},
	\label{eq:obs2d}
\end{align}
for $i=1,2$, with $w_{t,i}\sim\mathcal{N}(0,\sigma_w^2)$,
$v_{t,i}\sim\mathcal{N}(0,\sigma_v^2)$,
$\sigma_w^2=10$, $\sigma_v^2=1$.

\smallskip
\noindent\textbf{Candidate structures.}
As in Experiment~\ref{sec:exp1}, we consider
$\mathcal{S}=\{s_{\mathrm{nl}},\,s_{\mathrm{lin}}\}$.
Under $s_{\mathrm{nl}}$, the transition 
matches~\eqref{eq:dyn2d} exactly.
Under $s_{\mathrm{lin}}$, a linear transition
$z_{t+1,i} = \tfrac{1}{2}z_{t,i} + w_{t,i}$ is used 
with the same quadratic observation 
model~\eqref{eq:obs2d}, inducing structural mismatch 
in the transition component only.
Each structure-conditioned belief is propagated via a 
bootstrap particle filter~\cite{Gordon1993Bootstrap} with 
$N_p=1000$ particles per run.

\smallskip
\noindent\textbf{Implementation.}
Three methods are compared over $T=200$ steps and 
$M=100$ independent Monte~Carlo runs, all initialised 
at $s_0=s_{\mathrm{lin}}$ to test structural recovery:
\texttt{Fixed~LIN} ($s_t\equiv s_{\mathrm{lin}}$),
\texttt{Fixed~NL} ($s_t\equiv s_{\mathrm{nl}}$),
and CF with hysteresis margin $\delta=2.0$ and 
$W=10$-step windowed scores.
The larger margin relative to 
Experiments~\ref{sec:exp1}--\ref{sec:exp3} reflects 
the higher score variance that arises when $\Phi_t(s)$ 
accumulates contributions from both observation 
dimensions; the choice is consistent with
Corollary~\ref{cor:eventual_fixed_filtering}, which 
requires only that $\delta>0$ be calibrated against 
the score fluctuations induced by particle 
approximation (see Remark~\ref{rem:margin_2d}).

\begin{rem}[Score and margin in higher dimensions]
	\label{rem:margin_2d}
	In the two-dimensional setting, $\Phi_t(s)$ 
	accumulates contributions from both observation 
	dimensions, resulting in larger absolute values 
	and higher variance than in the scalar case.
	To suppress particle-induced score noise, a 
	$W=10$-step windowed average
	$\bar{\Phi}_t^W(s):=
	W^{-1}\sum_{k=0}^{W-1}\Phi_{t-k}(s)$
	is applied before the hysteresis check, and the 
	margin is set to $\delta=2.0$. Both choices are 
	consistent with 
	Corollary~\ref{cor:eventual_fixed_filtering}.
\end{rem}

\smallskip
\noindent\textbf{Results.}
Figure~\ref{fig:exp44} shows that CF identifies the
correct structure after a single early transition
($\rho_{\mathrm{sw}}=0.005$), after which the score
ordering $\Phi_t(s_{\mathrm{nl}})<\Phi_t(s_{\mathrm{lin}})$
is maintained and no further switching is triggered.
In both dimensions, CF recovers the accuracy of
\texttt{Fixed~NL}, consistent with the one-step
structural descent property
(Lemma~\ref{lem:struct_descent}) and the
finite-switching guarantee
(Corollary~\ref{cor:eventual_fixed_filtering}).
Quantitative metrics are reported in
Table~\ref{tab:all_experiments}.

These results confirm that the one-step structural 
descent property (Lemma~\ref{lem:struct_descent}) and 
the finite-switching guarantee
(Corollary~\ref{cor:eventual_fixed_filtering}) hold 
in the two-dimensional setting without any modification 
to the CF rule or its theoretical analysis.
Quantitative metrics are reported in 
Table~\ref{tab:all_experiments}, and establish that 
the three-layer theoretical framework of 
Section~\ref{sec:CFDeeepSSSM} generalises to 
multi-dimensional latent spaces as predicted.


\section{Conclusion}
\label{sec:conclusion}

We introduced cognitive flexibility (CF), a belief-level 
mechanism for online latent-structure selection in Bayesian 
filtering under structural mismatch. By selecting at each 
step the structure that minimises an innovation--based 
predictive score --- without modifying the underlying Bayesian 
recursion --- CF is well posed, exhibits a structural descent 
property, and reduces to standard filtering when a 
predictively consistent structure is available. Experiments 
across mismatch, shift, and well-specified regimes confirm 
that CF adapts only when necessary, switches finitely, and 
introduces no overhead under correct specification.
The irreducibility result (Theorem~\ref{thm:wellposedness_cf}) carries an immediate 
control-theoretic consequence: structural mismatch produces 
persistent degradation that parameter adaptation alone cannot 
correct. CF addresses this at the belief level, complementing 
robust and adaptive MPC frameworks~\cite{Hewing2020,Aswani2013} 
that assume fixed internal representations. Extending CF to 
closed-loop settings where the belief feeds directly into a 
control policy is a natural next step.





\bibliographystyle{elsarticle-num} 
\bibliography{Auto_Automatica} 

\end{document}